\documentclass{sig-alternate-2013}
\usepackage{array,graphicx,epsfig,amssymb,amsfonts,bbm,amsmath,epstopdf,algorithm,algorithmic,subfigure,cite,tabularx}
\usepackage{latexsym}
\usepackage{hyperref}
\usepackage{color}

\newcolumntype{L}[1]{>{\raggedright\let\newline\\\arraybackslash\hspace{0pt}}m{#1}}
\newcolumntype{C}[1]{>{\centering\let\newline\\\arraybackslash\hspace{0pt}}m{#1}}
\newcolumntype{R}[1]{>{\raggedleft\let\newline\\\arraybackslash\hspace{0pt}}m{#1}}

\newtheorem{theorem}{\textbf{Theorem}}
\newtheorem{lemma}{\textbf{Lemma}}
\newtheorem{proposition}{Proposition}

\newcommand{\myPrb}{{economic dispatching} }
\newcommand{\paonline}{\textsf{PA-Online} }
\newcommand{\poonline}{\textsf{PO-Online} }
\newcommand{\offline}{\textsf{OFFLINE} }
\newcommand{\mG}{{microgrid} }
\newcommand{\myScl}{0.20}

\newcommand{\algc}{\textsf{NRBF} }

\begin{document}

\conferenceinfo{ACM e-Energy}{'15 Bangalore, India}

\title{Peak-Aware Online Economic Dispatching for Microgrids}

\numberofauthors{5} %  in this sample file, there are a *total*
\author{
% 1st. author
\alignauthor
Ying Zhang\\
       \affaddr{Information Engineering}\\
       \affaddr{The Chinese University of Hong Kong}
%       \email{\small zy013@ie.cuhk.edu.hk}
% 2nd. author
\alignauthor
Mohammad H. Hajiesmaili\\
       \affaddr{Institute of Network Coding}\\
       \affaddr{The Chinese University of Hong Kong}
%       \email{\small mohammad@ie.cuhk.edu.hk}
% 3rd. author
\alignauthor
Sinan Cai\\
       \affaddr{School of Electrical Engineering}\\
       \affaddr{Xi'an Jiao Tong University}     
\and
\alignauthor
Minghua Chen\\
       \affaddr{Information Engineering}\\
       \affaddr{The Chinese University of Hong Kong}
\alignauthor 
Qi Zhu\\
       \affaddr{Department of Electrical and Computer Engineering}
       \affaddr{University of California, Riverside}
%       \email{\small minghua@ie.cuhk.edu.hk}
}

\maketitle

\begin{abstract}
\sloppy{By employing local renewable energy sources and power generation units
while connected to the central grid, microgrid can usher in great
benefits in terms of cost efficiency, power reliability, and environmental
awareness.
Economic dispatching is a central problem in microgrid operation, which aims at effectively scheduling various energy sources to minimize the operating cost while satisfying the electricity demand. Designing intelligent economic dispatching strategies for microgrids, however, is drastically different from that for conventional central grids, due to two unique challenges. First, the erratic renewable energy emphasizes the need for online algorithms. Second, the widely-adopted peak-based pricing scheme brings out the need for new peak-aware strategy design. In this paper, we tackle these critical challenges and devise peak-aware online economic dispatching algorithms. For microgrids with fast-responding generators, we prove that our deterministic and randomized algorithms achieve the best possible competitive ratios
$2-\beta$ and $e/(e-1+\beta)$, respectively, where $\beta\in[0,1]$ is the ratio between the minimum grid spot price and the local-generation price.
Our results characterize the fundamental \emph{price of uncertainty} of the problem. For microgrids with slow-responding generators, we first show that a large competitive ratio is inevitable. Then we leverage limited prediction of electricity demand and renewable generation to improve the competitiveness of the algorithms. By extensive empirical evaluations using real-world traces, we show that our online algorithms achieve near offline-optimal performance. In a representative scenario, our algorithm achieves $23\%$ and $11\%$ cost reduction as compared to the case without local generation units and the case using peak-oblivious algorithms, respectively.
}
\end{abstract}

%\vspace{-0.8\baselineskip}
%A category including the fourth, optional field follows...
%\vspace{-1\baselineskip}
%\terms{Algorithm, Design, Experimentation, Performance}
%\vspace{-1\baselineskip}

\section{Introduction}\label{sec:intro}
Microgrid represents a promising paradigm of future electric power systems that autonomously coordinate distributed renewable energy source (\textit{e.g.}, solar PVs), local generation unit (\textit{e.g.}, gas generators), and the external grid
to satisfy time-varying energy demand of a local community.
%
%In microgrid~\cite{lasseter2004microgrid}, distributed renewable energy source (\textit{e.g.},
%solar PVs), local generation unit (\textit{e.g.}, gas generators), and the external grid, are autonomously orchestrated
%to satisfy the time-varying energy demand of a local community.
%
%several local energy facilities, such as renewable energy sources (\textit{e.g.},
%solar PVs) and power generation units (\textit{e.g.}, gas generators), and the external grid, are autonomously orchestrated
%to satisfy the dynamic energy demand of a local community, such as university, corporate headquarter, and data center.
%By employing local energy sources to be able to operate in ``island mode'', %while still connected to the central grid,
As compared to traditional grids, microgrid has recognized advantages in cost efficiency, environmental awareness, and power reliability. Consequently, worldwide installed microgrid capacity has witnessed a phenomenon growth, reaching 866 MW in 2014, and is expected to reach 4,100 MW by 2020~\cite{2023microgrid}.

%Even though rapid penetration of renewable local sources provides environment-friendly energy landscape, uncertainty in availability of renewable sources results the fast variations of its output power. These fluctuations can affect both the economic and the reliability incentives of migrating toward microgrid. The key to mitigate this potential unpleasant effect is to deploy local generation units and and \textit{intelligently schedule} both local and external generations so as to minimize the \mG cost and to improve the power quality and reliability, in the meantime. For instance, several local generators could be deployed in a university campus to absorb the renewable uncertainty and to contribute in energy generation for economical and reliability benefits.

Energy generation scheduling in microgrid determines the power output
level of local generation units and power to be procured from external grid, with the goal of
minimizing the total cost over a pre-determined billing cycle. The scheduling
plan should meet the time-varying energy demand and respect physical
constraints of the generation units. Such problem has been studied extensively
in the power system literature for large-scale traditional grids. Two main variants are unit commitment
\cite{kazarlis1996genetic} and economic dispatching \cite{gaing2003particle}
problems. The unit commitment problem typically optimizes the start-up
and the shut-down schedule of power generation units, whereas the economic
dispatching problem optimally schedules the output levels given the
on/off status as the input parameters. In this paper, we focus on
economic dispatching problem in microgrid scenarios.

At first glance, \myPrb in \mG may appear to be a small-scale version
of the classical urban-wide \myPrb problem. However, the following two unprecedented challenges make the problem fundamentally different, thereby the previous solutions inapplicable.

$\vartriangleright$ \textbf{Uncontrollable, intermittent, and uncertain energy sources.} Classical scheduling strategies rely on accurate prediction of future demand and dispatch-able supply \cite{gaing2003particle}. In microgrids, however, the renewable sources are highly uncontrollable (not available on-demand), intermittent (irregular fluctuations), and uncertain (hard to predict accurately). Incorporating a large fraction of such renewable energy sources makes conventional strategies not applicable, and calls for new \textit{online} scheduling strategies that do not rely on accurate prediction of demand and renewable generation~\cite{narayanaswamy2012online,minghua_sigmetrics}.
%Another level of uncertainty is
%injected by considering co-generation
%capabilities. Observations in \cite{minghua_sigmetrics} corroborate
%that heat and electricity demands exhibits different patterns, and
%hence makes the prediction of overall energy demand even more acute.

$\vartriangleright$ \textbf{Peak-based charging model of the external
grid.} The real-world pricing scheme for consumers with large loads (such as universities or data centers) adopts
a hybrid time-of-use and \textit{peak-based} charging model where
the electricity bill consists of both the total energy usage and the
peak demand drawn over the billing cycle. The motivation is to
encourage large customers to smooth their demand, thereby the utility provider can reduce its planned capacity obligations. The peak price is
%significantly higher than the spot price.
often more than 100 times higher than the maximum (on-peak) spot price, \textit{e.g.}, $118$ times for PG\&E~\cite{pge}, and $227$ times for Duke Energy Kentucky~\cite{duke}~\footnote{In practice, the unit of peak price is \$/KW while the unit of spot price is \$/KWh. This estimation is obtained by assuming the peak demand lasts one hour.}.
Consequently, the contribution
of peak charge in the electricity bill for a typical costumer can be considerable, \textit{e.g.}, from 20\% to 80\% for several Google data centers \cite{xu2013reducing}. These observations
suggest that economic dispatching strategies with peak-based charging model taken
into account (referred to as peak-aware economic dispatching) may substantially
reduce the total operating costs for microgrids as compared to economic dispatching
strategies oblivious to peak-based charging (referred to as peak-oblivious
economic dispatching). This is indeed the case as verified by our real-world trace-driven
evaluation in Sec.~\ref{sec:exp}. %As an example based on real electricity demand traces of a college in San Francisco in July (YING) 2002 and electricity rates of PG\&E, the total cost (including monthly electricity bill and generator costs) of peak-aware \myPrb is $1549.2k\$$, while this vaule is $1377.4\$$ for peak-oblivious \myPrb, so another $10\%$ cost reduction happens compared with the original cost without local generators.

\iffalse
\begin{figure}
\label{fig:netdemand} \centering \includegraphics[bb = 0 0 200 100, draft, type=eps]{fig/trace}
\protect\caption{Hourly net power demand for July 2000}
\end{figure}

\begin{table}[!h]
\begin{tabular}{|c|c|c|}
\hline
algorithm  & cost($k$)  & cost reduction($\%$) w.r.t \textbf{only grid}\tabularnewline
\hline
only grid  & $1795.4$  & n.a. \tabularnewline
\hline
Peak-Obvious  & $1549.2$  & $23.28$\tabularnewline
\hline
Peak-Aware  & $1377.4$  & $13.70$\tabularnewline
\hline
\end{tabular}
\end{table}
\fi

All previous researches on \mG economic dispatching, that we are aware of and
review in Sec.~\ref{sec:relatedwork}, adopt a peak-oblivious cost
model, wherein the costumer bill is computed by total energy usage
following a time-of-use pricing scheme. To the best of our knowledge,
this work is the first that addresses the peak-aware \myPrb problem
using competitive online algorithms in microgrid scenario.
The main contributions of the paper are summarized as follows:

%\begin{itemize}
%\item
$\vartriangleright$ We identify and formulate the peak-aware economic dispatching problem of minimizing
the operating cost for microgrids under the hybrid time-of-use
and peak-based pricing scheme in Sec.~\ref{sec:formulation}. Notably, two aforementioned challenges change the structure of the problem fundamentally (see the discussions in Sec.~\ref{part:criticalCapa} for an example) and call for different online algorithm design.

$\vartriangleright$ In Sec.~\ref{sec:fast}, we focus on ``fast-responding'' generator scenario, where the \textit{ramping} constraints (\textit{i.e.,} the maximum change in output level over successive steps) of local generators are inactive.
%local generators have high ramping up/down rates thus the ramping constraints are inactive,
We follow a \textit{divide-and-conquer} approach and decompose the problem into multiple sub-problems, solve the sub-problems by their ``rent-or-buy'' nature, and then combine the solutions to obtain a solution for the original problem. We then demonstrate that the competitive ratios of our algorithms are $\left(2-\beta\right)$ and $e/\left(e-1+\beta\right)$ for deterministic and randomized versions respectively,
%$\frac{e}{e-\left(1-p_e^{\min}/p_g\right)^+}$,
%$\frac{e}{e-\left(1-p_e^{\min}/p_g\right)^+}$,
 where $\beta\in \left[0,1\right]$ is the ratio between the minimum grid spot price and the generator price. We prove that the ratios are the best possible. As such, these results characterize the fundamental \emph{price of uncertainty} for the problem.

%This structural decomposition allows us to solve these rent-or-buy problems independently and combine the solutions to obtain a solution for the original peak-aware \myPrb problem. Leveraging this critical insight, we devise a deterministic online algorithm for the peak-aware \myPrb problem, with a competitive ratio of $1+\left(1-p_{e}^{\min}/p_{g}\right)^{+}$, where $p_{e}^{\min}$ is the minimum grid electricity price and $p_{g}$ is the unit-output cost of local generators. We prove that the ratios are the best possible; thus, it characterizes the fundamental \emph{price of uncertainty} for this problem.
 %for the peak-aware \myPrb problem in the fast-responding generator scenario.

%\item
$\vartriangleright$ For ``slow-responding'' generator scenario in Sec.~\ref{sec:slow}, where the ramping constraints are active, we firstly show that a large competitive ratio is inevitable without any future information. We then design an online algorithm with a small competitive ratio by taking the advantage of sufficient looking-ahead information.
%, which corresponds to predicting renewable generation and local demand $3$ hours ahead in our empirical evaluation in (Sec.~\ref{part:ramp}).
Our results suggest looking-head as a useful mechanism to \textit{\emph{neutralize the ramping constraints in online algorithm design}}.

%\item
$\vartriangleright$ In Sec.~\ref{sec:exp}, by extensive evaluations using real-world traces, we show that our online algorithms can achieve satisfactory empirical performance. Furthermore, our \emph{peak-aware} online algorithms achieve near offline-optimal performance, and outperform the \textit{peak-oblivious} designs \cite{narayanaswamy2012online,minghua_sigmetrics} under various settings. %In a representative scenario, our algorithms achieve 23\% and 11\% cost reduction as compared to the case without local generators and the case with peak-oblivious algorithms, respectively.
The substantial cost reduction shows the benefit and necessity of designing peak-aware strategies for economic dispatching in microgrids.

\section{Problem Formulation}\label{sec:formulation}
\sloppy{
In the microgrid \myPrb problem, the objective is to orchestrate various energy sources to minimize the operating cost while satisfying the electricity demand.
%Those constraints include capacity limitation and the maximum rates of increase and decrease in generators' output level over successive steps, which is known as \textit{ramping} rates.
%The main challenges are the highly volatile net energy demands and the hybrid time-of-use and peak-based pricing scheme.
}

We consider one billing cycle, which is a finite time horizon set ${\mathcal{T} = \{1,\dots,T\}}$ with $T$ discrete time slots. In practice, the duration of one cycle is usually one month and the length of each time slot is $15$ minutes~\cite{pge}. The key notations used in this paper are defined in Table~\ref{tbl:not}. %, we proceed to introduce the main elements of the system model.
\begin{table}[!htp]
\caption{Key notations}\label{tbl:not}
\centering
\begin{tabular}{|c|L{6cm}|}
\hline \textbf{Notation} & \textbf{Definition} \\
\hline $T$ & The total number of time slots \\
 $\mathcal{T}$ & The time slot set\\
 \hline\hline
 $e(t)$ & The net electricity demand \\
 $u(t)$ & The electricity level obtained from local generators \\
 $v(t)$ & The electricity level obtained from electricity grid \\
\hline\hline $p_e(t)$ & The spot price of the electricity from grid at time $t, p_e^{\min} \leq p_e(t) \leq p_e^{\max}$,  (\$/\textrm{KWh})\\
 $p_g$ & The unit cost of the electricity by local generators (\$/\textrm{KWh})\\
 $p_m$ & The peak demand price of the electricity grid (\$/\textrm{KWh})\\
\hline \hline
 $R^\textsf{u}$ & The maximum ramping up rate of local generator\\
 $R^\textsf{d}$ & The maximum ramping down rate of local generator \\
  $C$ & Local generator capacity \\
\hline
\end{tabular}
\end{table}

\textbf{Net electricity demand.} We consider arbitrary renewable energy generation. Let $e(t)$ be the net electricity demand in time slot $t$, \textit{i.e.}, the total electricity demand subtracted by the renewable generation. For ease of presentation and discussion, we assume $e(t)$ only takes nonnegative integer values. Note that we do not assume any specific stochastic model of $e(t)$.

\textbf{Local generation.} There are local generators deployed in the microgrid with $C$ total generation capacity, \emph{i.e.}, they can jointly satisfy at most $C$ units of electricity demand. We consider a practical setting where the generator's incremental power output in two consecutive slots is limited by the \textit{ramping-up} and \textit{ramping-down} constraints $R^\textsf{u}$ and $R^\textsf{d}$, respectively. Most microgrids today employ small-capacity generators that are powered by gas turbines or diesel engines. These generators are ``fast-responding'' in the sense that they have large ramping-up/-down rates. Meanwhile, there are also ``slow-responding'' generators with small ramping-up/-down rates. We denote $p_g$ as the cost of generating unit electricity using local generation.

\textbf{Electricity from the external grid.} The microgrid can also obtain electricity supply from the external
grid for unbalanced electricity demand in an on-demand manner. We denote the spot price at time $t$ from the external
grid as $p_e(t)$. We assume that $p_e(t)\geq p_e^{\min} \geq 0$~\footnote{We remark that the electricity spot price can sometime be negative in practice  \cite{fanone2013case}. %The algorithms in this paper also apply to the scenarios with negative price. But
We restrict our attention to the case with $p_e(t)\geq 0$ in this study and leave the general case with negative price to future work.}. Again,
we do not assume any stochastic model of $p_e(t)$. For ease of discussion later, we define $\beta \triangleq p_e^{\min} / p_g$ as the ratio between the minimum grid price and the unit cost of local generation.

\textbf{Cost model.} The microgrid operating cost in $\mathcal{T}$ includes the expense of purchasing electricity from the external grid and that of local generation. Let $v(t)$ be the amount of electricity purchased from the external grid and $u(t)$ be the amount of electricity generated locally. %The net demand of a microgrid can be satisfied by external grid or local generation. As optimization variables,

The cost of grid electricity consists of volume charge and peak charge.
%Let $p_e(t)$ be the spot electricity price at time $t$, which is time varying and is lower bounded by $p_e^{\min}\geq 0$~\footnote{We remark that the electricity spot price can sometime be negative in practice  \cite{fanone2013case}. %The algorithms in this paper also apply to the scenarios with negative price. But
%We restrict our attention to the case with $p_e(t)\geq 0$ in this study and leave the general case with negative price to future work.}.
The volume charge is simply the sum of volume cost in all the time slots, \textit{i.e.}, $\sum_t p_e(t)v(t)$.
%Let $p_m$ be the peak price and the peak charge can be calculated by the peak demand over the billing cycle multiplied by $p_m$, \textit{i.e.}, $p_m\max_tv(t)$.
In practice, the peak charge is based on the maximum single-slot power and the peak price unit is \$/\textrm{KW}~\cite{pge}, which is different from the spot price unit \$/\textrm{KWh}.
%In this paper, since the variable $v(t)$ is the amount of purchased electricity in \textrm{KWh}, we convert the unit of $p_m$ to \$/\textrm{KWh}.
Let the peak price in $\$/\textrm{KW}$ be $\tilde{p}_m$ and the length of one time slot be $\delta$ (\textit{e.g.}, $0.25$ hour), we convert the peak price to \$/\textrm{KWh} as $p_m = \tilde{p}_m/\delta$. Consequently, the peak charge is $p_m\max_tv(t)$, \textit{i.e.}, the peak demand over the billing cycle (in \textrm{KWh}) multiplied by $p_m$ (in \$/\textrm{KWh}). This method is similar to the one used in \cite{xu2013reducing}. We remark that $p_m$ is usually more than 100 times larger than $p_e(t)$~\cite{pge}.

%According to some real electricity bills, the peak charge is based on the maximum power and the price is in \$/\textrm{KW}, which is different from the unit of spot price \$/\textrm{KWh}. In this paper, since $v(t)$ is the amount of purchased electricity in \textrm{KWh}, we convert the unit of $p_m$ to \$/\textrm{KWh}.\footnote{An alternative approach is to change the unit of $p_e(t)$ to \$/\textrm{KW} by considering the optimization variables as the power of each time slot \cite{xu2013reducing}}. We remark that $p_m$ is usually more than 100 times of $p_e(t)$.

For local generation, the cost of a generator to generate $\theta$ amount of electricity is commonly modeled as a quadratic function~\cite{kazarlis1996genetic}, \emph{i.e.}, say, $a\theta^2+b\theta+c$. The coefficient $a$ is usually orders of magnitude smaller than $b$ (\textit{e.g.}, for a typical oil generator with capacity 15MW, $a=0.007, b=48.5$) \footnote{This can be further verified by more examples from \scriptsize{http://pscal.ece.gatech.edu/archive/testsys/generators.html}.}. Consequently, for small-capacity generators employed in microgrids, the quadratic term $a\theta^2$ is usually much smaller than the linear term $b\theta$ and is negligible.
Let $p_g$ be the unit generating cost. The total local generation cost is simply $\sum_t p_g u(t)$. In this study, we focus on the case where $p_g\geq p_e(t)$, $\forall t\in \mathcal{T}$ \footnote{It means generating one unit of electricity locally is no cheaper than purchasing it from the external grid. The approaches and results developed in this paper can be extended to the general case where $p_g$ can be lower than $p_e(t)$.}.

%Otherwise, for the slots when $p_g <p_e(t)$, it would have been optimal to always obtain power by local generation as much as possible and satisfy the unbalanced electricity by the external grid.

%Consequently, we drop the quadratic term and follow the linear cost model in \cite{minghua_sigmetrics}.
%We focus on the economic dispatching problem, where the on/off status of generators are given as inputs by solving the unit commitment problem at larger time scale. %Thus the startup cost and sunk cost are beyond the scope of this paper.
%Meanwhile, the on/off status of the local generator is determined by solving the unit commitment problem. Since we focus on the economic dispatching problem, the startup cost and sunk cost are beyond the scope of this paper. thus the startup cost and sunk cost are beyond the scope of this paper, which focuses on the economic dispatching problem.

Putting together all the components, the microgird total operating cost over a billing cycle is given by
%\vspace{-0.2\baselineskip}
\begin{equation}
\label{eq:cost}
 \textsf{Cost}(\boldsymbol u, \boldsymbol v) = \underbrace{\sum_{t\in \mathcal{T}} p_e(t) v(t) + p_m \max_{t\in \mathcal{T}} v(t)}_{\textrm{by external grid}}+\underbrace{\sum_{t\in \mathcal{T}}p_g u(t)}_{\textrm{by local generators}}.
\end{equation}

Existing microgrid generation scheduling schemes~\cite{narayanaswamy2012online,minghua_sigmetrics} did not consider the peak charge term $p_m\max_tv(t)$; we refer to these schemes as \textbf{Peak-Oblivious}. In this paper, we consider the \textbf{Peak-Aware Economic Dispatching} (\textbf{PAED}) problem as follows
\begin{subequations}
\begin{eqnarray}
\textbf{PAED}\quad \min_{\boldsymbol u, \boldsymbol v} &&\textsf{Cost}(\boldsymbol u, \boldsymbol v)
\nonumber\\
\textrm{s.t.}
& & u(t)+v(t)\geq e(t), \quad t\in\mathcal{T},\label{equ:e_demand}\\
& & u(t)\leq C,  \quad t\in\mathcal{T}, \label{equ:cap}\\
& & u(t+1)-u(t)\leq R^\textsf{u}, \quad t\in\mathcal{T},\label{equ:ramping_u}\\
& & u(t)-u(t+1)\leq R^\textsf{d}, \quad t\in\mathcal{T},\label{equ:ramping_d}\\
\textrm{var.} & & u(t), v(t) \in \mathbb{R}^{+}, \quad  t\in\mathcal{T}.  \nonumber
\end{eqnarray}
\end{subequations}
The constraint in~\eqref{equ:e_demand} ensures that the electricity demand is satisfied. The constraint in~\eqref{equ:cap} is due to the generator capacity limitation. The constraints in~\eqref{equ:ramping_u}-\eqref{equ:ramping_d} reflect the ramping up/down constraints, respectively.

The objective function $\textsf{Cost}(\boldsymbol u, \boldsymbol v)$ is convex
and all the constraints are linear; hence \textbf{PAED} is a convex optimization problem. In the offline setting where the net demand in the entire time horizon, \emph{i.e.}, $e(t)$ for all $t$ in $\mathcal{T}$, is given (by for example accurate prediction), problem \textbf{PAED} can be solved easily using standard solvers. However, the net demand $e(t)$ in microgrid is hard to predict accurately as it inherits substantial uncertainty from renewable generation. This motivates the need of online strategies that do not rely on accurate net demand prediction to operate~\cite{minghua_sigmetrics}.

Denote an online algorithm for  problem \textbf{PAED} by $\mathcal{A}$, we use competitive ratio (\textbf{CR}) as the metric to evaluate its performance. For an online algorithm, its competitive ratio is defined as the maximum ratio between the cost it incurs and the offline optimal cost over all inputs, \textit{i.e.},
$$
    \textbf{CR} (\mathcal{A})\triangleq \max_{\textrm{all inputs}} \frac{\textsf{Cost incurred by } \mathcal{A}}{\textsf{Offline optimal cost}}.
$$
%For an optimization problem \textbf{P} and an online algorithm $\mathcal{A}$, we denote the
%\begin{equation*}
%$$\textbf{CR} (\mathcal{A})\triangleq \max_{\textrm{input}} \frac{\textsf{Cost}_{\mathcal{A}}}{\textsf{Cost}_{\textrm{offlineoptimal}}}.$$
%\end{equation*}
Clearly we have $\textbf{CR} \geq 1$. It is desired to design online algorithms with small competitive ratios, since it guarantees that, for any input, the cost of the online algorithm is close to the offline optimal. The \emph{price of uncertainty} (\textbf{PoU}) for  problem \textbf{PAED}  is defined as the minimum possible competitive ratio across all online algorithms, \emph{i.e.},
\[
    \textbf{PoU} \triangleq \min_{\textrm{all }\mathcal{A}} \textbf{CR} (\mathcal{A}).
\]

\section{Fast-Responding Generator Case}\label{sec:fast}

In this section, we relax the ramping constraints
\eqref{equ:ramping_u}-\eqref{equ:ramping_d} and consider the fast-responding
generator scenario. Most generators employed in microgrids can ramp
up/down very fast. For example, a diesel-based engine can ramp up/down
$40\%$ of its capacity per minute~\cite{vuorinen2007planning}.
Considering the time scale of each slot (\emph{e.g.}, 15 minutes),
those generators can be thought as having no ramping constraints.
That is, $R^{\textsf{u}}=R^{\textsf{d}}=\infty$. We note that even
though we relax the ramping constraints, the relaxed problem, denoted
as \textbf{FS-PAED}, still covers many practical scenarios in the
current microgrids \cite{minghua_sigmetrics}. Moreover, the results
in this section serves a building block for designing online algorithm
for the original problem \textbf{PAED} with ramping constraints, which
we will present in Sec.~\ref{sec:slow}.

%To keep the problem interesting, we assume that $p_g\geq p_e(t)$ for all $t\in \mathcal{T}$. Otherwise, for the slots where $p_g <p_e(t)$, it would have been optimal to always obtain power by local generation as much as possible and satisfy the unbalanced electricity by the external grid.

In the following, we first focus on a special version of problem $\textbf{FS-PAED}$,
named as $\textbf{FS-PAED}^{k}$, where the net demand only takes
value 0 or 1. We design optimal online algorithms for problem $\textbf{FS-PAED}^{k}$
and then extend the algorithms to solve the general problem $\textbf{FS-PAED}$.

%\sloppy{Note that when the grid spot price is higher than that of local generator, \emph{i.e.}, ${p_e(t)\geq p_g}$, it is straightforward to verify that it is optimal to first use the local generators to satisfy the demand as much as possible, \emph{i.e.}, ${u(t) = \min\{e(t),C\}}$, and then satisfy the residual demand using the external grid. %Hence in the following, we only focus on the nontrivial case where $p_e(t) < p_g,\forall t$.%And also, for more clear illustration, we assume the electricity demand $e(t)$ is discrete.\footnote{It means we cannot generate or buy infinitesimally small amount of electricity and it is reasonable considering practical scenarios.}

\subsection{Problem $\textbf{FS-PAED}^{k}$ and An Optimal Offline Solution\label{part:structure_off}}

We now consider a special version of problem $\textbf{FS-PAED}$ as
follows:
\begin{eqnarray*}
\textbf{FS-PAED}^{k}: & \min & \textsf{Cost}(\boldsymbol{u}^{k},\boldsymbol{v}^{k})\\
 & \textrm{s.t.} & u^{k}(t)+v^{k}(t)\geq e^{k}(t),\quad t\in\mathcal{T},\\
 & \textrm{var.} & u^{k}(t),v^{k}(t)\in\mathbb{R}^{+},\quad t\in\mathcal{T},
\end{eqnarray*}
where $e^{k}(t)$ only takes value 0 or 1.

We first study the offline setting, where the net demand $e^{k}(t)$,
$t\in\mathcal{T}$, is given ahead of time. We will reveal a useful
structure of the optimal offline solution, which we exploit to design
efficient online algorithms. Note that problem $\textbf{FS-PAED}^{k}$
can be solved by dynamic programming, which however does not seem
to bring significant insights for developing online algorithms. As
such, in what follows, we study the offline optimal solution from
another angle to reveal a useful structure.

Under the setting, the unit cost of local generation is more expensive
than the spot price of the external grid, \emph{i.e.}, $p_{e}(t)<p_{g}$.
However, the expensive local generation can be leveraged to cut off
the peak demand satisfied by the external grid and thus the prohibited
peak charge from the external grid. Thus, the key in solving problem
$\textbf{FS-PAED}$ lies in balancing between the cost of using the
expensive local generation and the peak charge of using the external
grid. It turns out the optimal offline solution, as shown in Lemma~\ref{lemma:opt_k}, is developed by comparing the accumulated cost of using the
local generation and the peak charge and leveraging the special structure
of $e^{k}(t)$.

\begin{lemma}\label{lemma:opt_k} An optimal
offline solution of $\textbf{FS-PAED}^{k}$, denoted by $\left\{ \left(\left(u^{k}(t)\right)^{*},\left(v^{k}(t)\right)^{*}\right)\right\} _{\mathcal{T}}$,
only takes value 0 and 1 and is given by $\left(u^{k}(t)\right)^{*}=e^{k}(k)-\left(v^{k}(t)\right)^{*}$
and
\begin{itemize}
\item if $\sigma>1$, then $\left(v^{k}(t)\right)^{*}=e^{k}(t)$, for all
$t$ in $\mathcal{T}$,
\item otherwise $\left(v^{k}(t)\right)^{*}=0$, for all $t$ in $\mathcal{T}$.
\end{itemize}
Here $\sigma$ is a critical peak-demand threshold defined by
\begin{equation}
\sigma\triangleq\frac{1}{p_{m}}\left[\sum_{t\in\mathcal{T}}\left(p_{g}-p_{e}(t)\right)e^{k}(t)\right].\label{eq:def.critical.peak}
\end{equation}

\end{lemma}

\textbf{Remark:} (i) Given that $e^{k}(t)$ is binary, certain mathematical
derivation shows that it suffices to constrain the variables $u^{k}(t)$
and $v^{k}(t)$ to be 0 or 1, and there is no need to consider the
cases where they take fractional values. This greatly simplify the
offline solution. (ii) The optimal solution constructed in Lemma \ref{lemma:opt_k}
is computed given that the critical peak-demand threshold $\sigma$
is determined. Meanwhile, $\sigma$ can only be computed in the offline
setting where the net demand in the entire horizon is given, and it
turns out it is the sufficient statistics of the net demand for characterizing
the ratio between the cost of an online algorithm and the offline
optimal cost.

\subsection{Online Algorithms for Problem $\textbf{FS-PAED}^{k}$}\label{ssec:online.algos.FS-PAED_k}

%\label{part:online}
The challenge for the online algorithm comes
from the fact that it cannot determine the value of critical peak-demand
threshold $\sigma$ ahead of time. This brings out a dilemma in online
decision making: \textit{to suffer deficit of local generator and
bypass the peak charge or to pay for the peak and enjoy cheaper electricity
from the grid}. %
%\footnote{Readers familiar with online algorithm literature may notice that
%a similar \textit{rent or buy} dilemma also occurs in the classical
%ski rental problem \cite{karlin1988competitive}, in which the customer
%can spend unit cost to rent a ski for one time or choose to buy a
%ski with the cost $B$ so that he does not need to rent any more.
%Since the customer does not know how many times he will go to ski
%in the future, he faces the same dilemma as we do in designing online
%strategies for Problem $\textbf{FS-PAED}^{k}$. %
%}.
The most \textit{aggressive} strategy acquires electricity from
the grid from the very beginning, while the most \textit{conservative}
strategy uses local generation to satisfy all the net demands in the
entire horizon, to avoid the peak charge.

An important observation in online decision making for problem $\textbf{FS-PAED}^{k}$
is that after purchasing electricity from the grid once, meaning the
peak charge has already been paid (and will not be charged again during
the current billing cycle), the microgrid should continue to use the
cheap electricity from the grid until the end of the billing cycle.
It turns out that the key decision is to determine when to start to
pay the peak-charge premium and buy electricity from the grid.

To pursue online algorithms with minimum competitive ratio, it turns
out that it suffices to focus on online algorithms that switch from
local generation to grid electricity procurement when the accumulated
local generation deficit exceeds $s\cdot p_{m}$, where $s\in [0,\infty)$ is an algorithm-specific parameter. For deterministic algorithms, these are the ones
switching to grid electricity procurement at time $\tau$ that satisfies
the following condition for the first time in the entire horizon:
\[
\sum_{t=1}^{\tau}\left(p_{g}-p_{e}(t)\right)e^{k}(t)\geq s\cdot p_{m}.
\]
The most aggressive strategy discussed above corresponds to $s=0$,
and the most conservative one corresponds to $s=\infty$. Randomized
online algorithms can be then characterized by distributions of $s$.

\subsubsection{An Optimal Deterministic Online Algorithm} \label{ssec:det_online_subproblm_k}

For any deterministic online algorithm with parameter $s$, denoted
by $\mathcal{A}_{s}$, the following proposition characterizes the
ratio between its online cost and the offline optimal cost.

\begin{proposition}\label{proposition:ratio} The ratio between the
cost of a deterministic online algorithm with parameter $s$ and the
offline optimal cost, denoted by $h\left(\mathcal{A}_{s},\sigma\right)$,
is given by: \\
when $\sigma\leq1$,
\begin{equation}
h\left(\mathcal{A}_{s},\sigma\right)=\begin{cases}
1, & \mbox{if }s>\sigma,\\
1+\frac{1-\sigma+s}{\sigma}(1-\beta), & \mbox{otherwise;}
\end{cases}
\end{equation}
when $\sigma>1$,
\begin{equation}
h\left(\mathcal{A}_{s},\sigma\right)=\begin{cases}
1+\frac{(\sigma-1)(1-\beta)}{(\sigma-1)\beta+1}, & \mbox{if }s>\sigma,\\
1+\frac{s(1-\beta)}{(\sigma-1)\beta+1}, & \mbox{otherwise.}
\end{cases}
\end{equation}
The competitive ratio for $\mathcal{A}_{s}$ is then
\begin{equation}
\mathbf{CR}\left(\mathcal{A}_{s}\right)=\max_{\sigma}h\left(\mathcal{A}_{s},\sigma\right).\label{eq:CR_As}
\end{equation}

\end{proposition}

\begin{proof}
We denote the number of time slots with demand $1$ by $T$, and the number of time slot using the local generator before turning to the grid by $T^s$.

$\vartriangleright$ \textbf{Case 1:} ${\sigma \leq 1}$. The optimal offline solution is always using the local generator and the cost is $\text{Cost}_{\text{off}} = Tp_g$.

\hspace{0.5cm}$\vartriangleright$ \textbf{Case 1.1:} $s> \sigma$. In this case, the online algorithm will not turn to the grid before the input ends. Therefore, the online cost is exactly  the same as the   offline cost, thereby the ratio is $1$.

\hspace{0.5cm}$\vartriangleright$ \textbf{Case 1.2:} $s \leq \sigma$. It turns out that there is a critical time slot $T^s$ that for all  $1 \leq t\leq T^s$, the online algorithm uses the local generator and for  time slots $T^s < t \leq  T$, it turns to the grid, thereby we have ${\text{Cost}_{\text{on}} = T^sp_g+\sum_{t = T^s+1}^{t=T}p_e(t)+p_m}$.
%However, the offline algorithm satisfies the total demand for the whole time horizon from the local generator (${\text{Cost}_{\text{off}} = Tp_g}$).
Hence, we get the following ratio:
\begin{align*}
h(s,\sigma) &= \frac{T^sp_g+\sum_{t = T^s+1}^{t=T}p_e(t)+p_m}{Tp_g}\\
&=\frac{Tp_g-(T-T^s)p_g+\sum_{t = T^s+1}^{t=T}p_e(t)+p_m}{Tp_g}\\
&=1+\frac{\sum_{t=1}^{T^s}(p_g-p_e(t))-\sum_{t=1}^{t=T}(p_g-p_e(t))+p_m}{Tp_g}\\
&\leq1+\frac{s p_m-\sigma p_m+p_m}{Tp_g}\\
& = 1+(1-\sigma+s)\frac{p_m}{T p_g}\\
&\leq 1+\frac{1-\sigma+s}{\sigma}\frac{p_g-p_e^{\min}}{p_g}\\
&=1+\frac{1-\sigma+s}{\sigma}(1-\beta)
\end{align*}
The last inequality is due to the fact that ${\sigma p_m = \sum(p_g-p_e(t))\leq T(p_g-p_e^{\min})}$.

$\vartriangleright$ \textbf{Case 2:} $\sigma > 1$. The optimal offline solution is always acquiring the electricity from the grid and the cost is ${\text{Cost}_{\text{off}}=\sum_{t=1}^{T} p_e(t)+p_m}$.

\hspace{0.5cm}$\vartriangleright$ \textbf{Case 2.1:} $s> \sigma$, In this case, the online algorithm always uses the local generator and thus the online cost is ${\text{Cost}_{\text{on}} = Tp_g}$. Hence, the ratio is as follows:
\begin{align*}
h(s,\sigma) &= \frac{Tp_g}{\sum_{t=1}^{T} p_e(t)+p_m}\\
& = 1+\frac{\sum_{t=1}^{T}(p_g-p_e(t))-p_m}{\sum_{t=1}^{T} p_e(t)+p_m}\\
& \leq 1+\frac{(\sigma-1)p_m}{Tp_e^{\min}+p_m}\\
& = 1+\frac{(z-1)}{p_e^{\min}\frac{T}{{p_m}}+1}\\
& \leq 1+\frac{(\sigma-1)}{p_e^{\min}\frac{\sigma}{p_g-p_e^{\min}}+1}\\
& = 1+\frac{(\sigma-1)(p_g-p_e^{\min})}{(\sigma-1)p_e^{\min}+p_g}\\
& = 1+\frac{(\sigma-1)(1-\beta)}{(\sigma-1)\beta+1}
\end{align*}
where the last inequality is true since  $\sigma p_m \leq T(p_g-p_e^{\min})$.

\hspace{0.5cm}$\vartriangleright$ \textbf{Case 2.2:} $s\leq \sigma$. Like case 2.1 here we have  $T\geq T^s$. Therefore, the online algorithm  uses the local generator for the first $T^s$ time slots and turns to the grid afterwards. In this case, the online cost is ${\text{Cost}_{\text{on}} = T^sp_g+\sum_{t = T^s+1}^{t=T}p_e(t)+p_m}$, and the ratio is
\begin{align*}
h(s,\sigma) &= \frac{T^sp_g+\sum_{t = T^s+1}^{t=T}p_e(t)+p_m}{\sum_{t=1}^{T} p_e(t)+p_m}\\
& =1+ \frac{\sum_{t=1}^{T^s}(p_g-p_e(t))}{\sum_{t=1}^{T} p_e(t)+p_m}\\
&\leq 1+\frac{sp_m}{Tp_e^{\min}+p_m}\\%, \text{by }\sum_{t=1}^{T^s}(p_g-p_e(t)) \leq sp_m\\
%& = 1+\frac{s}{p_e^{\min}\frac{T}{p_m}+1}\\
& \leq 1+\frac{s}{p_e^{\min}\frac{\sigma}{p_g-p_e^{\min}}+1}\\
& = 1+\frac{s(p_g-p_e^{\min})}{(\sigma-1)p_e^{\min}+p_g}\\
& = 1+\frac{s(1-\beta)}{(\sigma-1)\beta+1}
\end{align*}

The proof is completed.
\end{proof}
Based on the above proposition, we can design the best deterministic
online algorithm by solving the following min-max optimization problem
\begin{equation}
\min_{s}\max_{\sigma}h\left(\mathcal{A}_{s},\sigma\right).\label{equ:best_deter}
\end{equation}
The problem is non-convex and thus challenging on the first sight.
However, given a deterministic online algorithm $\mathcal{A}_{s}$,
it turns out the worst cost ratio is obtained when $\sigma=s$, in
which case the online algorithm pays for the peak-charge premium but
there is no net demand to serve anymore. Thus we have
\begin{align*}
\max_{\sigma}h\left(\mathcal{A}_{s},\sigma\right) & =h\left(\mathcal{A}_{s},s\right)=\begin{cases}
1+\frac{1}{s}(1-\beta), & \mbox{if }s\leq1,\\
1+\frac{s(1-\beta)}{(s-1)\beta+1}, & \mbox{otherwise}.
\end{cases}
\end{align*}
Leveraging this observation, the problem in (\ref{equ:best_deter})
can be solved easily by studying the extreme points of the two functions
of $s$, and the optimal value is obtained when $s=1$. To visualize
how the competitive ratio varies as $s$ changes, we plot the competitive
ratio for different values of $s$ in Fig.~\ref{fig:CR_diffs} for
the case where $\beta=0.3$.
\begin{figure}
\centering
\includegraphics[width=0.8\columnwidth]{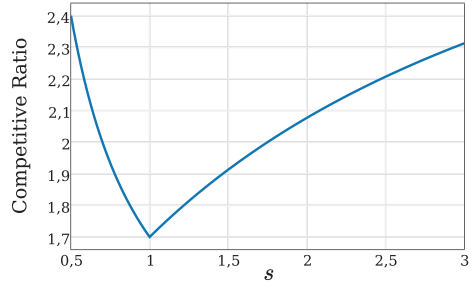}\\
 \protect\caption{Competitive ratio of $\mathcal{A}_{s}$ as a function of $s$, with
$\beta=0.3$.}
\label{fig:CR_diffs}
\end{figure}

We obtain the optimal deterministic online algorithm by setting $s=1$,
named as Break-Even Economic Dispatching for problem $\textbf{FS-PAED}^{k}$
(\textbf{BED-$k$}). The algorithm switches from local generation
to grid electricity procurement when the accumulated local generation
deficit seen so far just equals the peak charge, thus the name ``break-even
dispatching''. We summarize the algorithm \textbf{BED-$k$} into
Algorithm~\ref{alg:online_k}, and characterize its competitive
ratio in the following theorem.

\begin{theorem}\label{lemma:CR_k_deter} The competitive ratio of
\textbf{BED-$k$} is given by
\[
\mathbf{CR}\left(\mathbf{BED-}k\right)=2-\beta.
\]
% where $\beta=p_{e}^{\min}/p_{g}\in[0,1)$.
This also gives the price
of uncertainty suffered by all deterministic online algorithms, \emph{i.e.},
\[
\mathbf{PoU}_{\mbox{det}}=\min_{\mbox{all deterministic }\mathcal{A}_{s}}\mathbf{CR}\left(\mathcal{A}_{s}\right)=2-\beta.
\]
\end{theorem}

\begin{algorithm}[!ht]
\protect\caption{BED-$k$: Optimal deterministic online algorithm for $\textbf{FS-PAED}^{k}$}

\begin{algorithmic}[1]
\REQUIRE $p_{m}$,$p_{g}$,$p_{e}(t)$,$e^{k}(t)$
\ENSURE $u^{k}(t)$,$v^{k}(t)$
\STATE $\zeta=0$, $\tau=1$
\WHILE{$\tau\in\mathcal{T}$}
\IF{$\exists\iota<\tau$ such that $v^{k}(\iota)=1$}
\STATE $v^{k}(\tau)=e^{k}(\tau)$,
$u^{k}(\tau)=0$
\ELSE
\STATE $\zeta=\zeta+(p_{g}-p_{e}(\tau))e^{k}(\tau)$
\IF{$\zeta<p_{m}$}
\STATE $u^{k}(\tau)=e^{k}(\tau)$, $v^{k}(\tau)=0$
\ELSE
\STATE $v^{k}(\tau)=e^{k}(\tau)$, $u^{k}(\tau)=0$ \ENDIF
\ENDIF
\STATE $\tau=\tau+1$
\ENDWHILE
\end{algorithmic} \label{alg:online_k}
\end{algorithm}

We remark that the optimal deterministic algorithm is easy to implement
and achieves the minimum possible competitive ratio for problem $\textbf{FS-PAED}^{k}$.
Next, we proceed to design optimal randomized algorithm for the problem.

\subsubsection{An Optimal Randomized Online Algorithm}

Recall that for the purpose of designing randomized online algorithms
with the minimum competitive ratio for problem $\textbf{FS-PAED}^{k}$,
it suffices to consider algorithm $\mathcal{A}_{f}$ where $f$ represents
the probability distribution by which we generate the algorithm-specific
threshold $s$. Based on the analysis for deterministic online algorithms
in Sec. \ref{ssec:det_online_subproblm_k}, we can find the competitive
ratio of $\mathcal{A}_{f}$ by solving the following optimization
problem:
\begin{equation}
\mathbf{CR}\left(\mathcal{A}_{f}\right)=\max_{\sigma}\mathbf{E}_{f}\left[h\left(\mathcal{A}_{f},\sigma\right)\right]=\max_{\sigma}\int_{s}h\left(\mathcal{A}_{s},\sigma\right)f(s)ds.\label{eq:randomized.algo.CR}
\end{equation}

In the following, we first design a randomized online algorithm by
specifying a particular probability distribution and compute its competitive
ratio. We then leverage \emph{Yao's Principle} \cite{YaoPrinciple}
to obtain a lower bound of the competitive ratio of any randomized
algorithm. We will see the competitive ratio of our proposed online
algorithm matches the lower bound, establishing its optimality. The
result thus also characterizes the price of uncertainty suffered by
all randomized online algorithms.

We propose a randomized online algorithm by choosing the distribution
for $s$ as
\begin{equation}
f^{*}(s)=\begin{cases}
\frac{e^{s}}{e-1+\beta}, & \mbox{when }s\in[0,1];\\
\frac{\beta}{e-1+\beta}\delta(0), & \mbox{when }s=\infty;\\
0, & \mbox{otherwise}.
\end{cases}\label{equ:random_s}
\end{equation}
We summarize the resulting randomized online algorithm in to Algorithm
\ref{alg:random_online_k}, named as Randomized Economic Dispatching for problem $\textbf{FS-PAED}^{k}$
(\textbf{RED-$k$}). Its competitive ratio is characterized in the
following theorem.

\begin{theorem}\label{proposition:CR} With the distribution given
by $f^{*}(s)$ in \eqref{equ:random_s}, the competitive ratio of
\textbf{RED-$k$} is given by
\[
\mathbf{CR}\left(\mathbf{RED-}k\right)=\frac{e}{e-1+\beta}.
\]
% where $\beta=p_{e}^{\min}/p_{g}\in[0,1)$.

\end{theorem}

\begin{proof}
When $\sigma\leq 1$,
\begin{align*}
\int_sh(s,\sigma)f^*(s)ds & = \int_{0}^{\sigma} (1+\frac{1-\sigma+s}{\sigma}(1-\beta))\frac{e^s}{e-1+\beta} ds \\
&+ \int_{\sigma}^{1}\frac{e^s}{e-1+\beta} ds + \frac{e}{e-1+\beta}\\
& = 1+\int_0^{\sigma}\frac{1-\sigma+s}{\sigma}(1-\beta)\frac{e^s}{e-1+\beta}ds\\
& = 1+\frac{1-\beta}{e-1+\beta}\\
& = \frac{e}{e-1+\beta}
\end{align*}

When $\sigma>1$,

\begin{align*}
\int_sh(s,\sigma)f^*(s)ds & = \int_{0}^1 (1+\frac{1-\sigma+s}{\sigma}(1-\beta))\frac{e^s}{e-1+\beta} ds\\
 & + \int_1^{\sigma}\frac{e^s}{e-1+\beta} ds \\
 & +\int_{\sigma}^{+\infty}(1+\frac{(\sigma-1)(1-\beta)}{(\sigma-1)\beta+1})\frac{\beta}{e-1+\beta}\\
& = 1+\frac{1}{e-1+\beta}\frac{1-\beta}{(\sigma-1)\beta+1}(1+(\sigma-1)\beta)\\
& = 1+\frac{1-\beta}{e-1+\beta}\\
& = \frac{e}{e-1+\beta}
\end{align*}

Then $$\max_{\sigma}\int_sh(s,\sigma)f^*(s)ds = \frac{e}{e-1+\beta}.$$
The proof is completed.
\end{proof}
%PUT A PIECE OF PSEUDO CODE FOR RED-k here.

\begin{algorithm}[!ht]
\protect\caption{RED-$k$: Optimal randomized online algorithm for $\textbf{FS-PAED}^{k}$}

\begin{algorithmic}[1]
\REQUIRE $p_{m}$,$p_{g}$,$p_{e}(t)$,$e^{k}(t)$
\ENSURE $u^{k}(t)$,$v^{k}(t)$
\STATE generate $s$ according to the probability distribution specified in~(\ref{equ:random_s})
\STATE $\zeta=0$, $\tau=1$
\WHILE{$\tau\in\mathcal{T}$}
\IF{$\exists\iota<\tau$ such that $v^{k}(\iota)=1$}
\STATE $v^{k}(\tau)=e^{k}(\tau)$,
$u^{k}(\tau)=0$
\ELSE
\STATE $\zeta=\zeta+(p_{g}-p_{e}(\tau))e^{k}(\tau)$
\IF{$\zeta<s\cdot p_{m}$}
\STATE $u^{k}(\tau)=e^{k}(\tau)$, $v^{k}(\tau)=0$
\ELSE
\STATE $v^{k}(\tau)=e^{k}(\tau)$, $u^{k}(\tau)=0$ \ENDIF
\ENDIF
\STATE $\tau=\tau+1$
\ENDWHILE
\end{algorithmic} \label{alg:random_online_k}
\end{algorithm}

Now we leverage \textit{\emph{Yao's Principle}} \cite{YaoPrinciple}
to obtain a lower bound for the competitive ratio of any randomized
online algorithm. The idea is to choose a probability distribution
for $\sigma$, denoted by $g(\sigma)$, and compute the competitive
ratio of the best deterministic online algorithm for this input. Yao's
Principle says that the computed ratio is a lower bound for any randomized
online algorithm. The particular distribution we use is given by
\begin{equation}
g^{*}(\sigma)=\begin{cases}
\frac{e}{e-1+\beta}\sigma e^{-\sigma}, & \mbox{when }\sigma\in[0,1],\\
\frac{e}{e-1+\beta}[(\sigma-1)\beta+1]e^{-\sigma}, & \mbox{otherwise.}
\end{cases}\label{equ:random_z}
\end{equation}
The lower bound is characterized in the following lemma.

\begin{theorem}\label{proposition:lower_bound} For any randomized
online algorithm $\mathcal{A}_{f}$ for problem $\textbf{FS-PAED}^{k}$,
we have
\[
\mathbf{CR}\left(\mathcal{A}_{f}\right)\geq\frac{e}{e-1+\beta}.
\]
The competitive ratio of algorithm \textbf{RED-$k$} achieves this
lower bound and thus is optimal. Consequently, the price of uncertainty
suffered by all randomized online algorithms is given by
\[
\mathbf{PoU}_{\mbox{ran}}=\min_{\mbox{all randomized }\mathcal{A}_{f}}\mathbf{CR}\left(\mathcal{A}_{f}\right)=\frac{e}{e-1+\beta}.
\]
\end{theorem}

\begin{proof}
When $s\leq 1$,
\begin{align*}
\int_{\sigma}h(s,\sigma)g^*(\sigma)d\sigma & = \int_{0}^s g^*(\sigma)d\sigma + \int_{s}^{1}(1+\frac{(1-\sigma)s}{z})g^*(\sigma) d\sigma\\
 & + \int_{1}^{+\inf}(1+\frac{s(1-\beta)}{(\sigma-1)\beta+1})g^*(\sigma) d\sigma \\
& = 1+ \frac{e(1-\beta)}{e-1+\beta}\cdot \\
& \left[\int_s^1(1-\sigma+s)e^{-\sigma}dz + \int_{1}^{+\inf}se^{-\sigma}d\sigma\right]\\
& = 1+\frac{e(1-\beta)}{e-1+\beta}e^{-1}\\
& = \frac{e}{e-1+\beta}
\end{align*}

When $s > 1$,
\begin{align*}
\int_{\sigma}h(s,\sigma)g^*(\sigma)d\sigma & = \int_{0}^1 g^*(\sigma) d\sigma \\
& + \int_{1}^{s} (1+\frac{(\sigma-1)(1-\beta)}{(\sigma-1)\beta +1})g^*(\sigma)d\sigma \\
& + \int_{s}^{+\inf} (1+\frac{s(1-\beta)}{(\sigma-1)\beta+1})g^*(\sigma)d\sigma\\
& = 1+\frac{e(1-\beta)}{e-1+\beta}\cdot\\
& \left[\int_1^s (\sigma-1)e^{-\sigma} d\sigma+\int_s^{+\inf}se^{-\sigma} d\sigma\right]\\
& = 1+\frac{e(1-\beta)}{e-1+\beta}e^{-1}\\
& = \frac{e}{e-1+\beta}
\end{align*}

Then $$\min_{s}\left[\int_{\sigma}h(s,\sigma)g^*(\sigma)d\sigma\right] = \frac{e}{e-1+\beta}.$$
The proof is completed.
\end{proof}
\textbf{Remark: }(i) In the deterministic online algorithm, setting
$s=1$ means that the microgrid will start to buy electricity from
the grid until the break-even condition is met. Similar to the ski
rental problem \cite{karlin1988competitive}, the break-even point
turns out to be the best balance between being aggressive and conservative.
(ii) The vigilant readers may notice that $f^{*}(s)$ is the same
distribution that was adopted in solving the classic Bahncard problem
\cite{bahncard}, which is indeed similar to problem $\textbf{FS-PAED}^{k}$
we study in this section. The basic version of $\textbf{FS-PAED}^{k}$,
however, is different from Bahncard problem in the sense that the
\textit{discounted price} ($p_{e}(t)$ in this paper) is time varying.
(iii) Different from the neat tricks used in \cite{bahncard} to prove
the optimality of the proposed randomized online algorithm \cite{bahncard},
we leverage \textit{Yao's Principle} to prove the optimality of our
proposed algorithm \textbf{RED-$k$} for problem $\textbf{FS-PAED}^{k}$.
Exploiting the similarity of the two problems, our approach can also
be applied to establish optimality of the proposed algorithm for the
Bahncard problem in \cite{bahncard}.

\subsection{From Problem $\textbf{FS-PAED}^{k}$ to Problem $\textbf{FS-PAED}$}

\label{subsec:fromto} In this section, we design online deterministic
and randomized algorithms for $\textbf{FS-PAED}$ based on those of
$\textbf{FS-PAED}^{k}$.

\subsubsection{Net Demand Layering}

For each time slot $t$, we divide the demand $e(t)$ into multiple
layers such that the demand of each layer is either $1$ or $0$,
as shown in Fig.~\ref{fig:layering}. Recall that $e(t)$ is assumed
to take non-negative integer values. We denote the sub-problem of
satisfying the demand of each layer as $\textbf{FS-PAED}^{k},k=1,2,...$,
and we can apply the online algorithms \textbf{BED}-$k$ and \textbf{RED}-$k$
for each sub-problem.

\begin{figure}
\centering
\includegraphics[width=0.95\columnwidth]{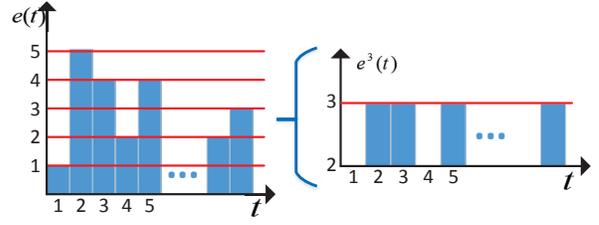}\\
 \protect\caption{An example of decomposing the demand into multiple layers and a microscopic
view of layer 3}

\label{fig:layering}
\end{figure}

\subsubsection{Optimal Online Algorithms for $\textbf{FS-PAED}$ }

After layering, a bunch of sub-problems $\textbf{FS-PAED}^{k}$ are
obtained. However, unlike $\textbf{FS-PAED}^{k}$, the net demand
of $\textbf{FS-PAED}$ in some time slots can exceed the capacity
of local generation, which makes it infeasible to ignore the whole
picture when conquering each layer independently. For example, suppose
the generation capacity is 4 for the case shown in Fig.~\ref{fig:layering}.
Even though the break even points are not reached for all the layers
in time slot 2, it is infeasible to set $u^{k}(2)=1$ for all the
layers (A capacity of 5 is needed to do so). Thus by taking into account
the capacity constraint, we need to determine for which layers the
demand should be satisfied by the grid while still keeping the algorithm
competitive.%and making the decisions for different layers as independent as possible.

An obvious but critical observation is that the demands in the lower
layers are denser than those in the upper layers. In addition, after
being charged for the peak, we expect more demands to come to enjoy
the cheap grid electricity. Consequently, it is always more economic
to use the grid electricity to satisfy the denser demands, \emph{i.e.},
the lower layers. In other words, in the proper algorithm design,
the layers below $(e(t)-C)^{+}$ should always be satisfied by the
grid. Meanwhile, for the layers above $(e(t)-C)^{+}$, if the demand
is already satisfied by the grid, the online algorithm continues to
acquire the electricity from the grid; otherwise, Algorithm \textbf{BED}-$k$
or\textbf{ RED}-$k$ is applied with the same value $s$ for all layers
to obtain the sub-solutions. The solution is finally obtained by combining
the sub-solutions. We summarize the resulting deterministic an randomized
online algorithms, named as \textbf{BED} and \textbf{RED}, in Algorithm~\ref{alg:deterministic.online}
and \ref{alg:randomized.online}, respectively.

\begin{algorithm}[!h]
\protect\caption{\textbf{BED}: Optimal deterministic online algorithm for $\textbf{FS-PAED}$}

\begin{algorithmic}[1] \REQUIRE $C$,$p_{m}$,$p_{g}$,$p_{e}(t)$,$e(t)$
\ENSURE $u(t)$,$v(t)$ \WHILE{$\tau\in\mathcal{T}$} \STATE
A threshold: $\varsigma=(e(\tau)-C)^{+}$. \STATE For the layers
below $\varsigma$, $v^{k}(\tau)=1$, $u^{k}(\tau)=0$ \STATE For
the layers above $\varsigma$, run \textbf{BED-$k$} to obtain $u^{k}(\tau)$
and $v^{k}(\tau)$. \STATE $u(\tau)=\sum_{k}u^{k}(\tau)$, $v(\tau)=\sum_{k}v^{k}(\tau)$
\STATE $\tau=\tau+1$ \ENDWHILE \end{algorithmic} \label{alg:deterministic.online}
\end{algorithm}

\begin{algorithm}[!h]
\protect\caption{\textbf{RED}: Optimal randomized online algorithm for $\textbf{FS-PAED}$}

\begin{algorithmic}[1] \REQUIRE $C$,$p_{m}$,$p_{g}$,$p_{e}(t)$,$e(t)$
\ENSURE $u(t)$,$v(t)$ \WHILE{$\tau\in\mathcal{T}$} \STATE
A threshold: $\varsigma=(e(\tau)-C)^{+}$. \STATE For the layers
below $\varsigma$, $v^{k}(\tau)=1$, $u^{k}(\tau)=0$ \STATE For
the layers above $\varsigma$, run \textbf{RED-$k$} with the same randomized parameter $s$
 to obtain $u^{k}(\tau)$
and $v^{k}(\tau)$. \STATE $u(\tau)=\sum_{k}u^{k}(\tau)$, $v(\tau)=\sum_{k}v^{k}(\tau)$
\STATE $\tau=\tau+1$ \ENDWHILE \end{algorithmic} \label{alg:randomized.online}
\end{algorithm}

We demonstrate a toy example of the solution given by \textbf{BED}
in Fig.~\ref{fig:layering2}. For simplicity, we assume the break-even
condition is firstly met when the third nonzero demands comes for
all the layers, and the local capacity is $4$. We use different colors
to demonstrate by which source and for what reason one unit of demand
is satisfied. Even though the example is simple, it demonstrates two
important and provable properties of \textbf{BED}: (i) For each layer,
it will continue to use the grid after it uses it once, and (ii) when
one layer uses the grid, all the layers below it use the grid too.
The first property makes the solution and cost structure similar to
that of \textbf{BED}-$k$, while the second property makes the peak
of $v(t)$ equal to the sum of the peaks of $v^{k}(t)$, \emph{i.e.},
$\max_{t}\sum_{k}v^{k}(t)=\sum_{k}\max_{t}v^{k}(t)$. The two properties
allow us to leverage the results in Sec. \ref{ssec:online.algos.FS-PAED_k}
to establish the competitive ratios of \textbf{BED} and \textbf{RED}.

\begin{figure}
\centering
\includegraphics[width=0.95\columnwidth]{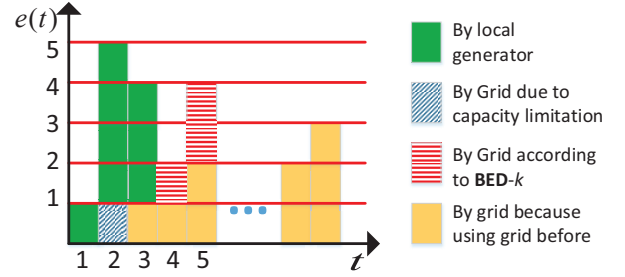}\\
 \protect\caption{Demonstration of \textbf{BED} with $C=4$, different colors denoting different
strategies of the algorithm.}

\label{fig:layering2}
\end{figure}

\begin{theorem}\label{theorem:CR_algs_FS-PAED} The competitive ratios
of \textbf{BED} and \textbf{RED} are given by
\[
\mathbf{CR}\left(\mathbf{RED}\right)=2-\beta,\;\;\;\mbox{and \;\;\;}\mathbf{CR}\left(\mathbf{RED}\right)=\frac{e}{e-1+\beta}.
\]
Further, no other deterministic and randomized online algorithm can
achieve a smaller competitive ratio.

\end{theorem}

\begin{proof}
%This proof applies for Theorem~\ref{theorem:CR_deter} and \ref{theorem:CR_random}.
Firstly, if the given input violates the capacity constraints, we can construct a new input respecting the capacity constraints, by sequently removing the demand exceeding the capacity and the following demands in the same layer (like the first layer from time slot 2 in Fig~\ref{fig:layering2}).  Then compared with the original input, the online cost and offline cost are reduced by the same amount, which leads to a larger competitive ratio. Then we only need to focus on the input whose demand is always smaller than the capacity.

Furthermore, due to the second property of the algorithm, we can have $\max_t\sum_kv^k(t) = \sum_k \max_t v^k(t)$, and $\begin{cases}u(t) = \sum_ku^k(t)\\v(t) = \sum_kv^k(t)\end{cases}$. Then $\textsf{Cost}(\boldsymbol u, \boldsymbol v) = \sum_k\textsf{Cost}({\boldsymbol u}^k, {\boldsymbol v}^k)$. This property still holds for the offline cost. We denote $\tilde{r}$ as the competitive ratio for each layer, meaning
$$\textsf{Cost}({\boldsymbol u}^k, {\boldsymbol v}^k) \leq \tilde{r}\text{Cost}^k_{\text{off}}, \forall k.$$

Then by summing the above inequality on $k$, we can have $$\textsf{Cost}(\boldsymbol u, \boldsymbol v) \leq \tilde{r} \text{Cost}_{\text{off}}.$$

For \textbf{BED}, $\tilde{r} = 2-\beta$ and for \textbf{RED}, $\tilde{r} = \frac{e}{e-1+\beta}$ for the randomized case, which establish the upper bound of the competitive ratios.

Furthermore, note that $\textbf{FS-PAED}^k$ is a special case of $\textbf{FS-PAED}$. Since we cannot obtain smaller competitive ratios for $\textbf{FS-PAED}^k$, we cannot obtain smaller competitive ratios for $\textbf{FS-PAED}$.

The proof is completed.
\end{proof}

In the next subsection, we discuss an intriguing consequence of local
generation capacity on the online algorithms' performance.

\subsection{Critical Local Generation Capacity}

\label{part:criticalCapa}

%For the time slot with $p_{e}(t)>p_{g}$, it is easy to see that boththe online and offline algorithms will benefit more from larger capacity since it is economic to use the cheap local electricity to its maximum. But when $p_{e}(t)\leq p_{g}$, things are different.

The peak-aware economic dispatching aims at minimizing the sum of
the peak charge (the term $p_{m}\max_{t\in\mathcal{T}}v(t)$ in~\eqref{eq:cost})
and the volume charge (as the remaining part in~\eqref{eq:cost}).
The local generator provides the microgrid an option to use more expensive
electricity (increase the volume charge) to reduce the peak (decrease
the peak charge). An optimal solution is achieved with the best tradeoff
between the two. Given an input, there is a threshold $\tilde{C}$,
the demand below which should be satisfied by the grid and above which
by the local generator. $\tilde{C}$ can be obtained by solving \textbf{FS-PAED}
in an offline fashion without considering capacity constraint. It
means that the optimal offline solution will not use the additional
capacity even if it is larger than $\tilde{C}$.

We now discuss the impact of increasing local generation capacity
$C$ on the performance of offline and online algorithms. The offline
algorithm will use full local capacity until $C$ reaches $\tilde{C}$,
and it will not use local capacity further beyond $\tilde{C}$. As
such, one can expect that the operating cost of the offline algorithm
is non-increasing as $C$ increases. Meanwhile, the online algorithm,
without knowing $\tilde{C}$ and with the tendency of reducing the
peak with more expensive electricity, will try to exploit the whole
capacity until it finds the break even point, which turns out to be
less economic and deviates more from the optimal solution. As a result,
for the online algorithm, larger capacity may incur higher operating
cost. %Taking the case in Fig.~\ref{fig:layering2} as a simple example, if the capacity increases from 4 to 5, the demand of time slot 2 in the first layer will be satisfied by the local generator while the others remains the same, increasing the total cost by $(p_g-p_e(2))$. Consequently, the performance of the online algorithm is degraded by a larger capacity.
We provide a concrete case-study by real world traces to confirm the
above observation in Sec.~\ref{sec:exp}.

Overall, we believe the above insights are important for microgrid
operators to (a) determine the amount of local generation to invest
in order to maximize the economic benefit, and (b) understand the
importance of demand/generation prediction when performing peak-aware
economic dispatching in microgrids. 
%\vspace{-1\baselineskip}
%\vspace{1cm}
\section{Slow-Responding Generator Case}\label{sec:slow}

This section considers the slow-responding generator
scenario, in which the ramping up/down constraints in \eqref{equ:ramping_u}-\eqref{equ:ramping_d}
are non-negligible. We remark that we can still optimally solve the
problem \textbf{PAED} with these constraints in the offline manner
by convex optimization techniques.

In the following analysis, we assume %
\mbox{%
$R^{\textsf{u}}=R^{\textsf{d}}=R$%
} and define $\Gamma=\lceil\frac{C}{R}\rceil$.
Then, it takes $\Gamma$ time slots for the local generator's output
to ramp up from zero to full capacity or down from full capacity to
zero. Considering the time scale of our problem (say, 15 minutes for
each time slot) and the microgrid scenario (high efficiency of local
generators), $\Gamma$ is conceivable to be small. We assume $\Gamma$
is no larger than 5, meaning it roughly takes no more than 75 minutes
for local generator to fully ramp up.

We first show a result highlighting the difficulty introduced by the
ramping constraints in designing competitive online economic dispatching
algorithms.

\begin{proposition} \label{prop:large_CR_with_ramping_constraints}

Any online algorithm for problem \textbf{PAED} without future information, \emph{i.e.}, at time $t$ the algorithm only have knowledge
of $\left\{ e(\tau),p_{e}(\tau)\right\} _{\tau=1}^{t}$, has a competitive
ratio at least $\ensuremath{\frac{p_{m}(C-R)+p_{g}R}{p_{g}(R\Gamma(\Gamma-1)+C)}}$.

\end{proposition}

When $\Gamma$ is $5$ and $p_{m}$ is 100 times of $p_{g}$, a back-of-envelop
calculation reveals that the lower bound of the competitive ratio
can be as large as $20$. This result shows that the ramping constraints
will make any online algorithm design less attractive as the worst
performance can be rather bad. The conventional method to address
this problem is to put additional constraints on the input to obtain algorithms with reasonable performance guarantee
\cite{narayanaswamy2012online}. In this paper, we propose to handle
the challenge incurred by ramping constraints by a different approach;
that is to empower the algorithm with a limited looking ahead window.

\subsection{An Effective Online Algorithm by a Limited Looking-ahead Window}

Motivated by the development of prediction algorithms \cite{taylor2006comparison,zhu2011case},
we assume that a limited looking ahead window with size of $\Delta=\Gamma-1$
is available, which means that at time $t$ we can know the input
from $t+1$ to $t+\Delta$ in advance. In this section, we devise
an online algorithm by leveraging such looking ahead information as
well as the results in Sec~\ref{sec:fast}. We name the proposed
algorithm as \algc (Neutralize Ramping constraint By Future information).
We denote the online solutions we obtain for \textbf{PAED} by \algc
as $\tilde{u}(t),\tilde{v}(t)$.

In \algc, we first solve problem \textbf{FS-PAED} by relaxing the
ramping constraints from problem $\textbf{PAED}$ and denote the solutions
obtained by algorithms \textbf{BED} or \textbf{RED} as $u(t)$ and
$v(t)$. We then adjust them to obtain online solutions for $\textbf{PAED}$
that satisfy the ramping constraints. Specifically, we compute $\tilde{u}(t)$
as
\[
\tilde{u}(t)=\max\{\tilde{u}(t-1)-R,u(t+i)-iR|i=0,1,..,\Delta\},
\]
and $\tilde{v}(t)=\max\left(e(t)-\tilde{u}(t),0\right)$.

The following lemma shows that the ramping constraints are respected
by \textsf{NRBF}.

\begin{lemma}\label{lemma:ramping_feasibility} The solutions by
\algc satisfy the generator's ramping constraints, \textit{i.e.},
\[
|\tilde{u}(t+1)-\tilde{u}(t)|\leq R.
\]
\end{lemma}

\begin{proof}
By the obtainment of $\tilde{u}(t+1)$, we can have $\tilde{u}(t+1)\geq \tilde{u}(t)-R$, i.e
$$\tilde{u}(t)-\tilde{u}(t+1)\leq R$$

Next, we prove $\tilde{u}(t+1)-\tilde{u}(t)\leq R$ for two cases.

Firstly, if $\tilde{u}(t+1) = (\tilde{u}(t)-R)^+$, the conclusion holds obviously.

Secondly, if for some $i^*\in[0,w-1]$ such that $\tilde{u}(t+1) = u(t+1+i^*)-i^*R$, we can have
\begin{align*}
\tilde{u}(t+1) &=\max\{0, u(t+1+i)-iR|i\in[0,\Delta]\}\\
&=\max\{0,u(t+1+i)-(i+1)R+R|i\in[0,\Delta]\} \\
&=\max\{0,u(t+j)-jR|j\in [1,\Gamma]\}+R\}\\
&=\max\{0,u(t+j)-jR|j\in [1,\Delta]\}+R\}
\end{align*}

The last step is due to the fact that $u(t+w)\leq C\leq wR$. and since $\tilde{u}(t) =\max\{0,\tilde{u}(t-1)-R, \max\{u(t+j)-jR|i\in [0,w-1]\}\}$, we can have $\tilde{u}(t)\geq \max\{0,u(t+j)-jR|j\in[0,\Delta]\}\geq \{0,u(t+j)-jR|j\in[1,\Delta]\}$.

Then it's easy to see that $\tilde{u}(t+1) \leq \tilde{u}(t)+R$, i.e
$\tilde{u}(t+1)-\tilde{u}(t)\leq R,$
which means the ramping up constraint is satisfied and the proof is completed.
\end{proof}

Meanwhile, we can have $\tilde{v}(t)\leq v(t)$, meaning the peak
charge is upper bounded by that of the fast responding scenario. We
leverage this observation to show the competitiveness of \textsf{NRBF},
as shown in Theorem~\ref{theorem:CR_ramp}.

\begin{theorem}\label{theorem:CR_ramp} The competitive ratio of
\algc satisfies
\[
\mathbf{CR}\left(\mbox{\algc}\right)\leq\begin{cases}
\Gamma\left(2-\beta\right), & \mbox{if }u(t),\, v(t)\mbox{ are obtained by \textbf{BED}};\\
\Gamma\frac{e}{e-1+\beta}, & \mbox{if }u(t),\, v(t)\mbox{ are obtained by \textbf{RED}}.
\end{cases}
\]
\end{theorem}
 %$\omega\left[1+\left(1-\frac{p_e^{\min}}{p_g}\right)^+\right]$\end{theorem}

\begin{proof}

According to the algorithm we can easily get
\begin{lemma}\label{lemma:ramping_cost_grid}
For $\tilde{u}(t)$,$\tilde{v}(t)$ obtain by \algc, we can have
$\tilde{v}(t)\leq v(t)$, and thus $\max_t\tilde{v}(t)\leq \max_t v(t)$
\end{lemma}

We denote the online cost of \textsf{NRBF} as $\tilde{C}_{\textrm{on}}$ and the optimal offline cost of the same problem as $\tilde{C}_{\textrm{off}}$. For the problem without ramping constraints, we denote the online cost of \textbf{BED} or \textbf{RED} as $C_{\textrm{on}}$ and the optimal offline cost of the same problem as $C_{off}$. We can directly have $\tilde{C}_{\textrm{off}}\geq C_{\textrm{off}}$. Then
\begin{align*}
\frac{\tilde{C}_{\textrm{on}}}{\tilde{C}_{\textrm{off}}} &\leq \frac{\tilde{C}_{\textrm{on}}}{C_{\textrm{off}}}\\
& = \frac{\tilde{C}_{\textrm{on}}}{C_{\textrm{on}}}*\frac{C_{\textrm{on}}}{C_{\textrm{off}}}\\
&\leq \frac{\tilde{C}_{\textrm{on}}}{C_{\textrm{on}}}*\mathcal{CR}
\end{align*}

If $\frac{\tilde{C}_{\textrm{on}}}{C_{\textrm{on}}}\leq 1$, this theorem holds obviously; otherwise, we denote $\begin{cases}\tilde{C}_{\textrm{on}} = \tilde{C}_{\textrm{on}}^l+\tilde{C}_{\textrm{on}}^g\\ C_{\textrm{on}} = C_{\textrm{on}}^l+C_{\textrm{on}}^g\end{cases}$, where the superscripts $l$ and $g$ represent the costs from local generators and external grid respectively.

By Lemma~\ref{lemma:ramping_cost_grid}, we can have $\tilde{C}_{\textrm{on}}^g\leq C_{\textrm{on}}^g$; then
\begin{equation*}
\frac{\tilde{C}_{\textrm{on}}}{C_{\textrm{on}}}\leq \frac{\tilde{C}_{\textrm{on}}^l}{C_{\textrm{on}}^l}
\end{equation*}

Next we obtain an upper bound for $\frac{\tilde{C}_{\textrm{on}}^l}{C_{\textrm{on}}^l}$ by modifying the input a little bit. We expand each time slot to one segment by inserting $\frac{u(t)}{R}$ intervals with $0$ demand before $t$ and another $\frac{u(t)}{R}$ intervals with $0$ demands after $a(n)$.
Then, the online cost stays the same for the problem without ramping constraints. Furthermore, the online cost for the problem with ramping constraints increases.

For the $m^{th}$ segment,$C_{\textrm{on}}^l(m) = p_gu(m)$, while $\tilde{C}_{\textrm{on}}^n(m) =p_g( \sum_{i=1}^{\Gamma-1}iR+u(m)+\sum_{i=1}^{\Gamma-1}iR)\leq p_g*u(m)\Gamma$; meaning for any $m$ we can have
$$\frac{\tilde{C}_{\textrm{on}}^l(m)}{C_{\textrm{on}}^n(m)}\leq \Gamma$$

Then
\begin{align*}
\frac{\tilde{C}_{\textrm{on}}^l}{C_{\textrm{on}}^l} &=\frac{\sum_m\tilde{C}_{\textrm{on}}^l(m)}{\sum_mC_{\textrm{on}}^l(m)}\\
&\leq \Gamma
\end{align*}

Further we can have $\frac{\tilde{C}_{on}}{\tilde{C}_{off}}\leq \frac{C}{R}\mathcal{CR}$, which completes the proof.

\end{proof}

\textbf{Remark}: Small values of $\Gamma$, which mean the ramping
constraints are less strict, will lead to a small bound on the competitive
ratio. Moreover, $\Gamma=1$ indicates that the generator output can
ramp up to its full capacity in one time slot and the ramping constraints
vanish, thereby we do not need any future information ($\Delta=0$)
and the competitive ratio is exactly the same with that of the fast-responding
generator scenario. 
\section{Experimental Results}\label{sec:exp}
We carry out numerical experiments using real-world traces to scrutinize the performance of our online algorithms under various practical settings. Our purpose is to investigate (i) the competitiveness of online algorithms in comparison with the optimal offline one, (ii) the necessity of peak-awareness in economic dispatching of microgrids, and (iii) the performance of online algorithms under various parameter settings.
%The observations in this part can provide a deep understanding for this peak-aware economic dispatching problem.

\begin{figure}[t]
\centering
\begin{minipage}[!htbp]{0.50\columnwidth}
\includegraphics[scale=\myScl]{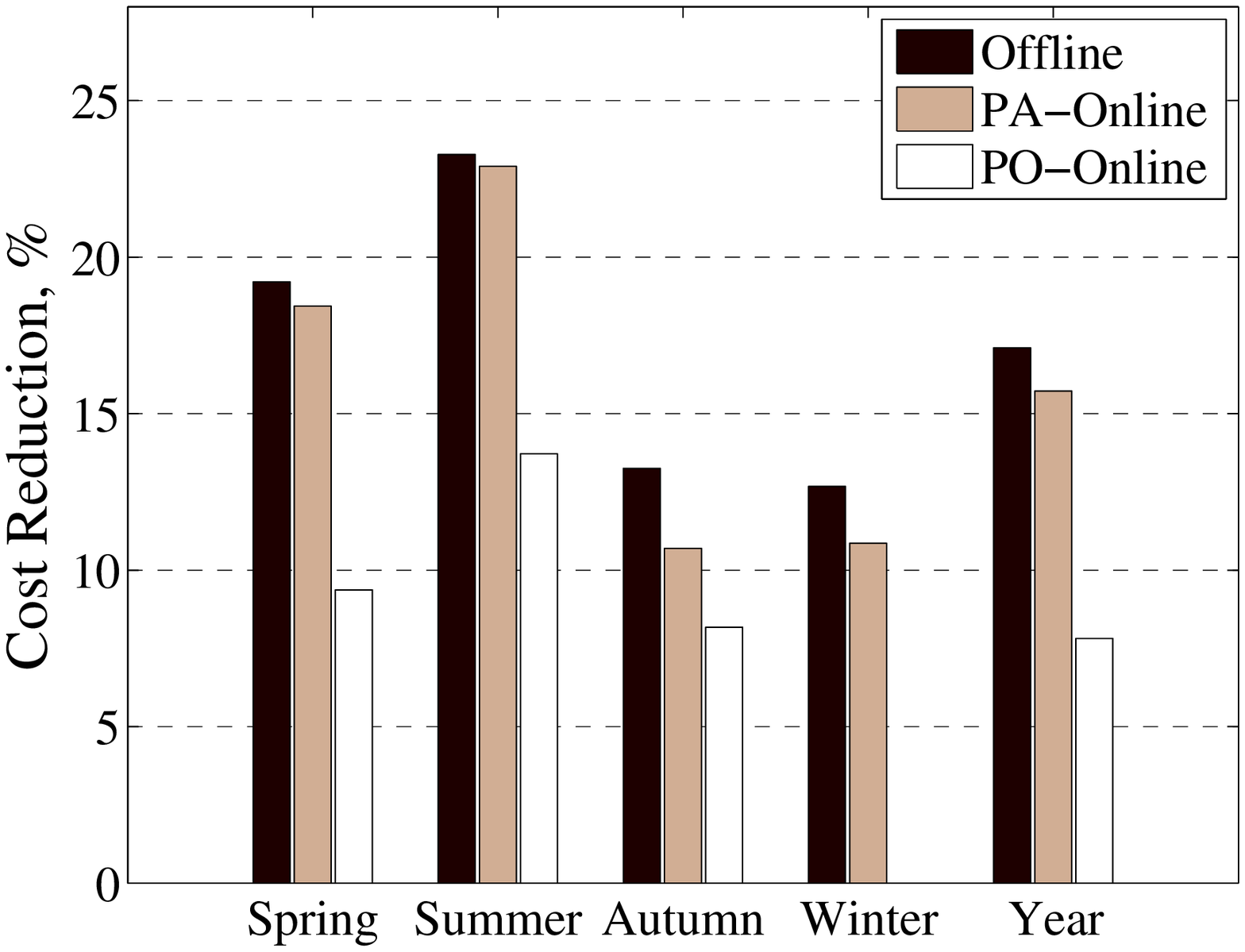}
\caption{Cost reduction for different seasons and the whole year}
\label{fig:RSeason}
\end{minipage}%
\begin{minipage}[!htbp]{0.50\columnwidth}
\includegraphics[scale=\myScl]{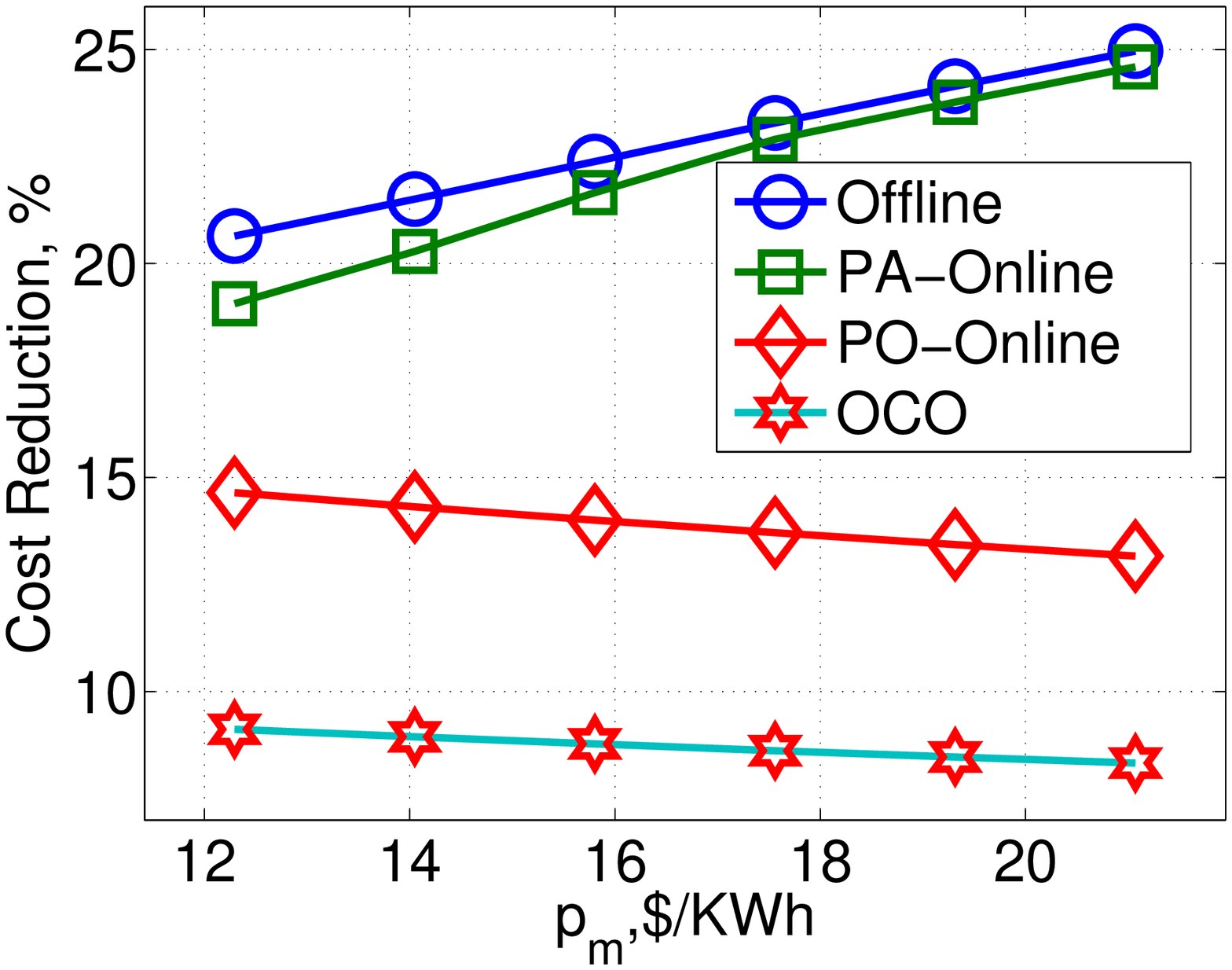}
\caption{Cost reduction vs $p_m$}
\label{fig:RPeakP}
\end{minipage}
\iffalse
\begin{minipage}[!htbp]{0.47\columnwidth}
\includegraphics[scale=\myScl]{ReductionLocalPrice.eps}
\caption{Cost reduction vs $p_g$}
\label{fig:RLocalPrice}
\end{minipage}%
\begin{minipage}[!htbp]{0.47\columnwidth}
\includegraphics[scale=\myScl]{CRLocalPrice.eps}
\caption{Competitive ratio vs $p_g$}
\label{fig:CRLocalPrice}
\end{minipage}
\fi
\end{figure}
\subsection{Experimental Setup}
\textbf{Electricity demand and renewable generation traces.} We set the length of one billing circle as one month. We use the actual electricity demand of a college in San Francisco; its yearly demand is about $154\textrm{GWh}$~\cite{college}. We inject renewable energy supply sources by a wind power trace of a nearby offshore wind station outside San Francisco with a total installed capacity of $12\textrm{MW}$ \cite{wind}. We then construct the net demand by subtracting the output level of the wind from the college electricity demand.

\textbf{Energy source parameters.} The electricity price $p_e(t)$ and peak price $p_m$ are set based on the tariffs from PG\&E\cite{pge} %and the values are mentioned in detail in motivating example in Sec.~\ref{sec:motex}.
and $p_m = 17.56\$/\textrm{KWh}$ while the electricity rate $p_e(t)$ varies from $0.056\$/\textrm{KWh}$ to $0.232\$/\textrm{KWh}$ for off-, mid-, and on-peak periods in different seasons. We set the unit cost of local generation $p_g$ according to the monthly price of natural gas. Notably, the value of $p_g$ could be less than $p_e(t)$ for some on-peak intervals. In such situations, generator plays its role not only by cutting off the peak but also by providing cheaper electricity as well. Finally, if not specified, the capacity of the local generator is set to be $C=15\textrm{MWh}$, which is around $60\%$ of the peak net demand.

\textbf{Cost benchmark.} We use the cost incurred by only procuring electricity from the external grid, \emph{i.e.}, $v(t)=e(t)$, as the benchmark. We demonstrate cost reduction to show the benefit of employing local generation units and the effectiveness of algorithms. The cost reduction originates from the cheaper electricity (in some on-peak intervals) and peak cut-off by local generators. %(Recall that the net demand is obtained by subtracting the wind generation from the actual college demand.)
%Different reductions by different algorithms demonstrate their effectiveness.%And the cost reductions can be different by different algorithms or under different parameter settings.

\textbf{Comparison of algorithms.} We compare our proposed peak-aware online \myPrb (\textsf{PA-Online}) algorithms  with (i) the optimal peak-aware offline solution (\textsf{OFFLINE}) to evaluate the performance of the online algorithms, and (ii) the peak-oblivious online algorithms (\textsf{PO-Online}) in \cite{minghua_sigmetrics} and online convex optimization approach (\textsf{OCO}) in \cite{narayanaswamy2012online} to investigate the importance of peak-awareness.\footnote{We remark that in \cite{minghua_sigmetrics}, the joint unit commitment and economic dispatching problem in peak-oblivious manner is addressed and in this paper we compare the economic dispatching part with our algorithms. \textsf{OCO} in \cite{narayanaswamy2012online} (without considering the peak charge) is designed deliberately to tackle the ramping constraints and may suffer performance loss in the fast responding generator scenario.} We remark that both schemes in ~\cite{minghua_sigmetrics,narayanaswamy2012online} are peak-oblivious as they only consider volume charge but ignore peak charge.
%{\color{blue} Note that the  algorithms proposed in \cite{minghua_sigmetrics} also take into account the start-up cost of generator. Considering the start-up cost in peak-oblivious scheduling is important since the only incentive of using  generator is to produce electricity in on-peak intervals. In this way, an intelligent scheduling should carefully calculate the start-up cost as well as the volume cost. But, in peak-aware scheduling the incentive is different; we are interested to acquire some portion of the demand from the generator even in mid- or off-peak intervals so as to cut off the grid peak demand. Consequently, the generator is on most of the times in peak-aware scheduling.}

The results reported in Secs.~\ref{part:season}-~\ref{part:capa} cover the fast-responding generator scenarios and Sec.~\ref{part:ramp} is devoted to the slow-responding generator scenario. %And all the algorithms we test are deterministic ones, while the corresponding randomized one should perform better on average.

\subsection{Benefits of Employing Local Generators}\label{part:season}
\textbf{Purpose.}
The purpose of this experiment is two-fold. First, compare the potential savings of microgrid in different seasons, in which the demand pattern, the wind output, and the cost parameters differ. Second, compare the cost reduction of peak-aware algorithms against peak-oblivious ones. The results are shown in Fig.~\ref{fig:RSeason}.

\textbf{Observations.}
The most notable observations from Fig.~\ref{fig:RSeason} are the following. First of all, the cost reduction varies over seasons and the most significant one occurs in the summer. This is because the gas price is lower and the grid electricity price is higher in the summer than those of the other seasons, thus employing local generators brings more benefit. Second, the performance of our proposed \paonline is superior than \poonline algorithm. In particular, \poonline cannot reduce the cost in the winter, but our algorithm \paonline can still achieve cost reduction.  The reason is that, as $p_g > p_e(t)$ always holds in the winter, \mbox{\poonline} algorithm always purchases cheaper electricity from the gird, which gives no cost reduction as compare to the benchmark strategy. In contrast, our \paonline algorithm reduces the cost by exploiting (the expensive) local generation to reduce the peak demand served by the external grid, and consequently our algorithm can save operating cost. On average, \paonline reduces the annual cost by $15.7\%$, while \poonline reduces the cost only by $8.17\%$. \textit{Third}, the performance of \paonline in practice is close to that of the offline optimal.

%It is better than the performance guarantee we provide in Theorem~\ref{theorem:CR_deter}. This is because the competitive ratio characterizes the performance in the worst case, which rarely happens in practice.

\begin{figure}[t]
\centering
\begin{minipage}[!htbp]{0.5\columnwidth}
\includegraphics[scale=\myScl]{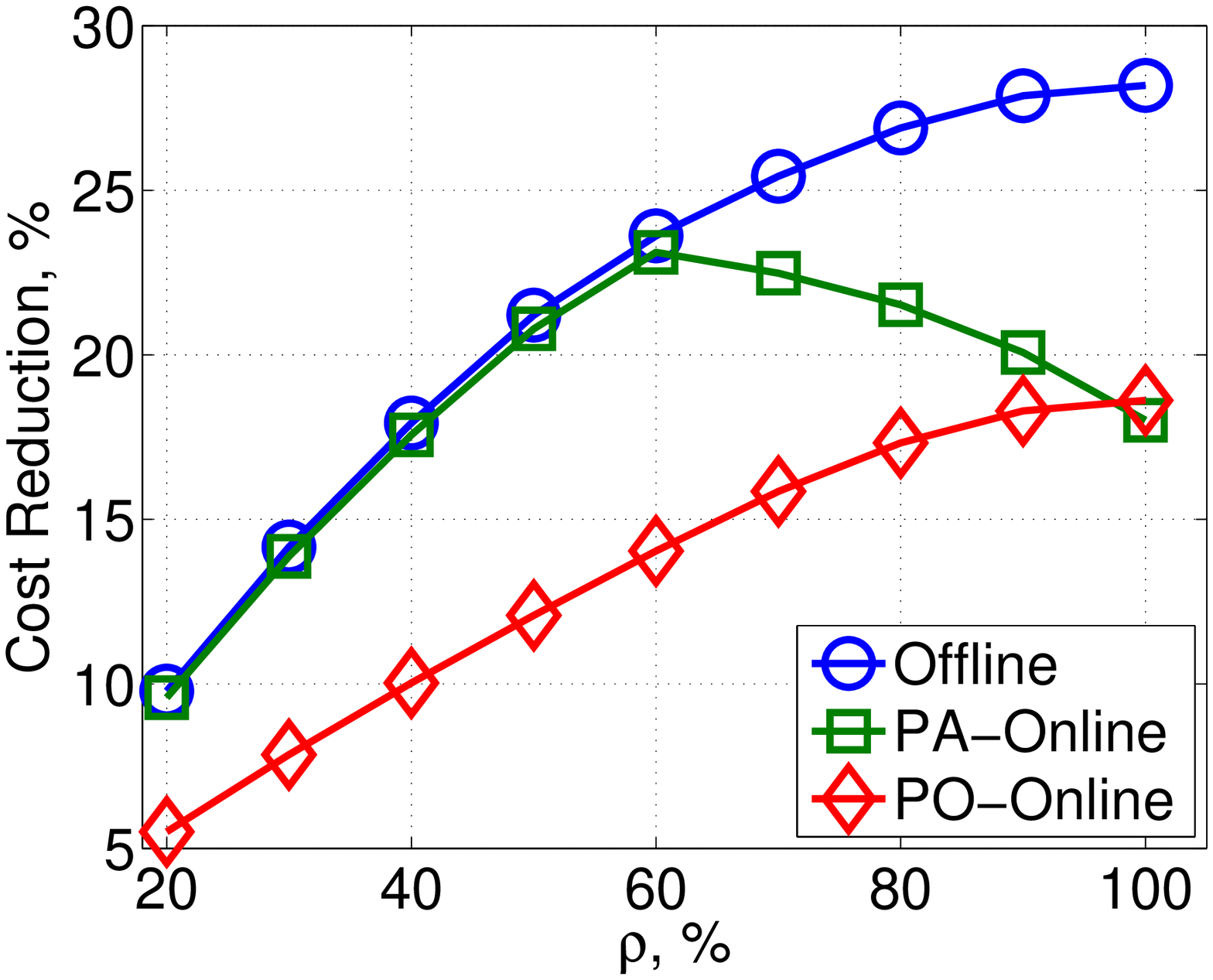}
\caption{Cost reduction vs $\rho$}
\label{fig:ReductionCapa}
\end{minipage}%
\begin{minipage}[!htbp]{0.5\columnwidth}
\includegraphics[scale=\myScl]{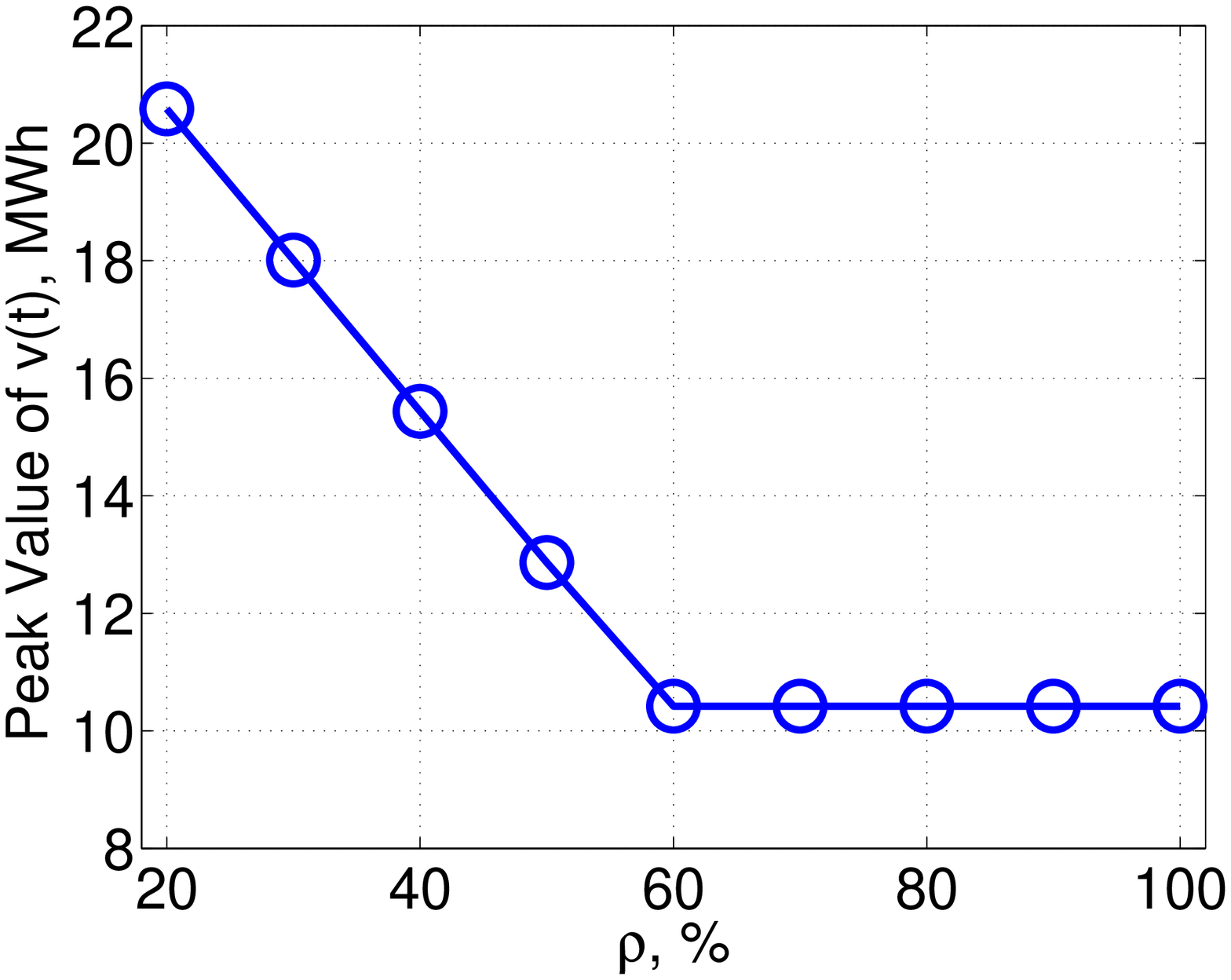}
\caption{Peak value of $v(t)$ vs $\rho$}
\label{fig:VPCapa}
\end{minipage}
\end{figure}

\begin{figure}
\begin{minipage}[!htbp]{0.50\columnwidth}
\includegraphics[scale=\myScl]{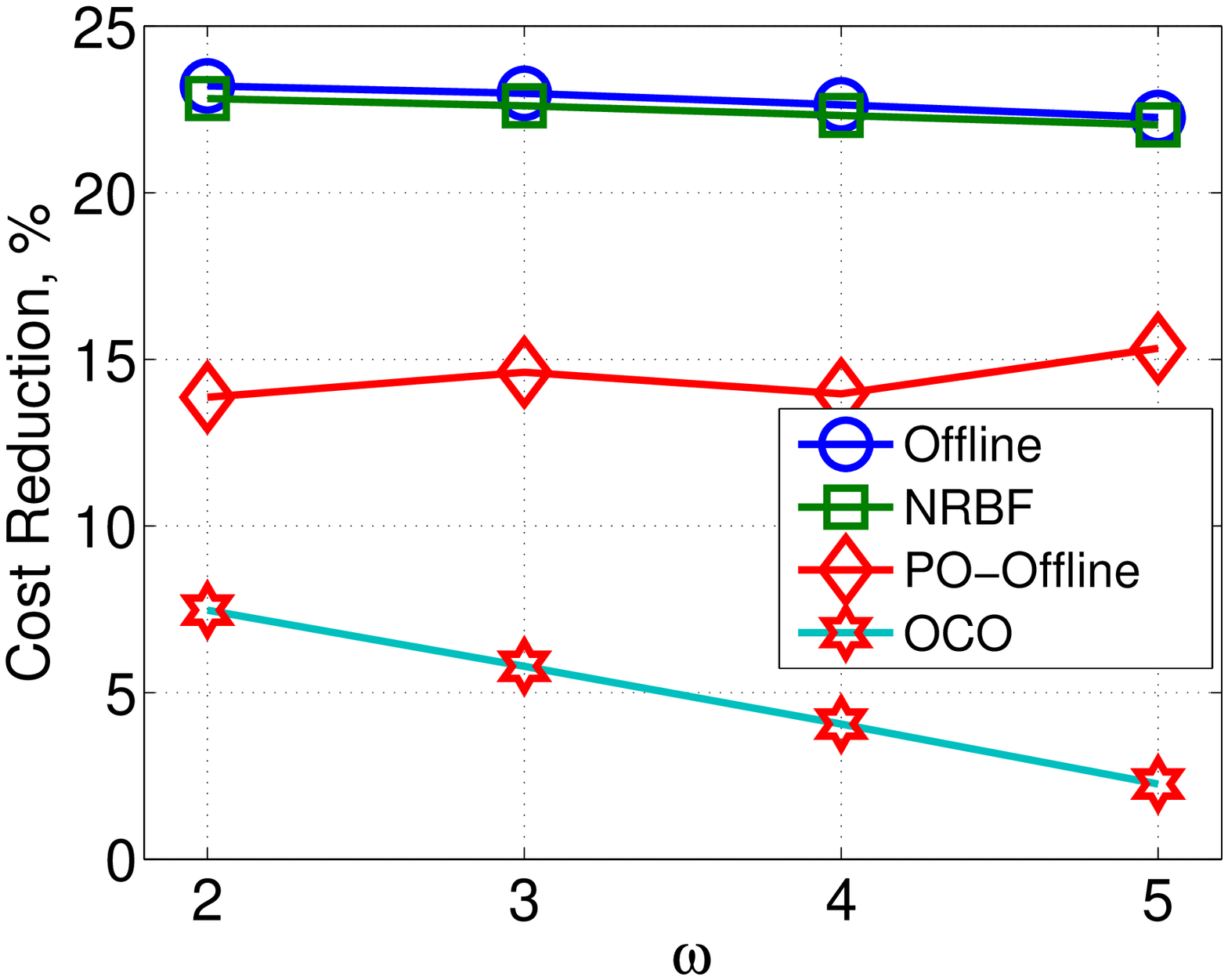}
\caption{Cost reduction vs $\Gamma$}
\label{fig:RRamp}
\end{minipage}%
\begin{minipage}[!htbp]{0.50\columnwidth}
\includegraphics[scale=\myScl]{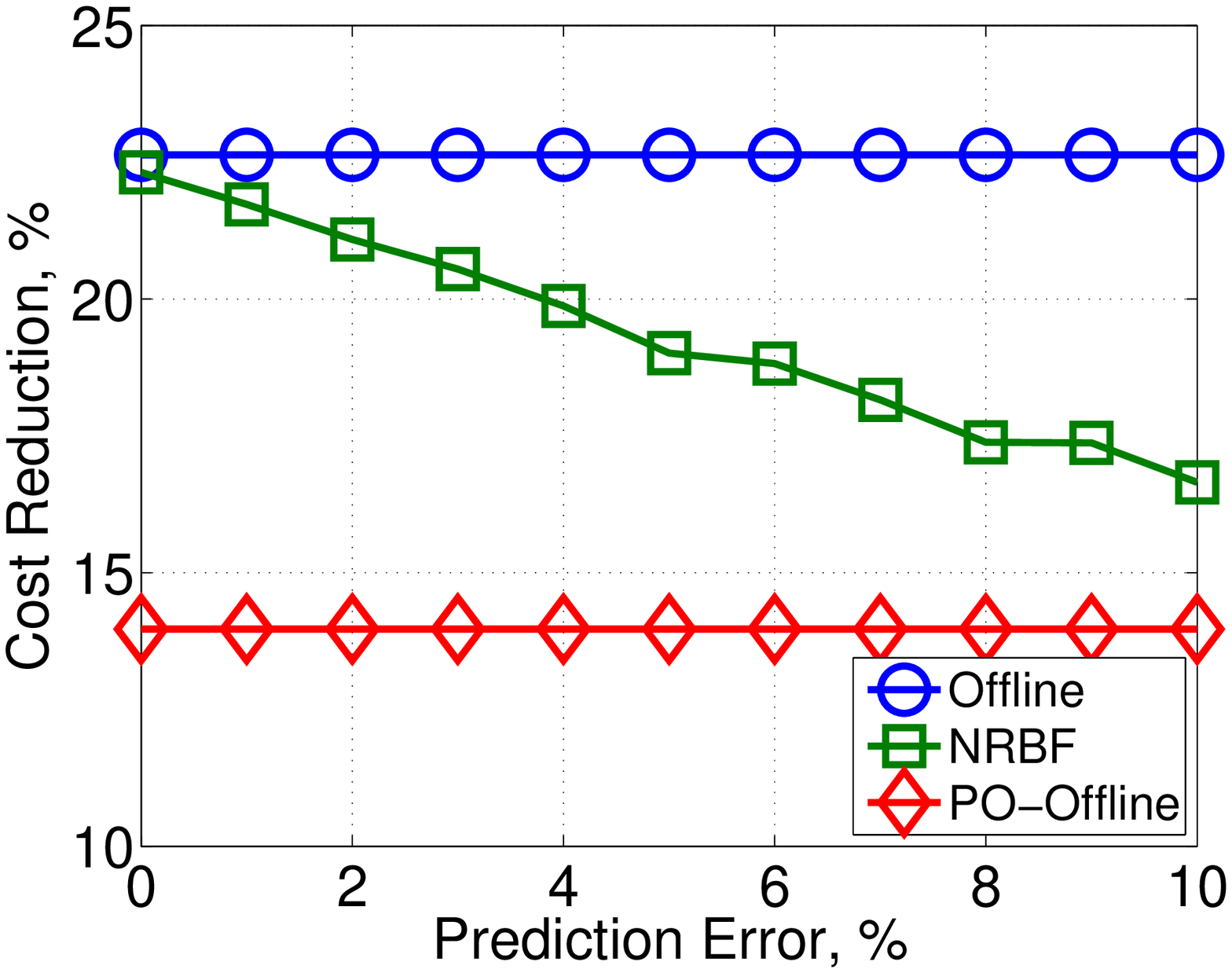}
\caption{Cost reduction vs prediction error}
\label{fig:CRRamp}
\end{minipage}
\end{figure}

\subsection{The Performance of \paonline under Different Peak Prices}
\textbf{Purpose.} To validate the peak charge is non-negligible which motivates our study, we evaluate the performance of our peak-aware algorithm and that of the peak-oblivious one under different peak prices. In particular, in Fig.~\ref{fig:RPeakP}, we depict the cost reduction of different algorithms with the peak price $p_m$ varying from $12.29\$/\textrm{KWh}$ to $21.07\$\textrm{/KWh}$.

\textbf{Observations.} When $p_m$ increases, the cost reduction of our \paonline algorithm increases and is close to the offline optimal, while the reductions of the two peak oblivious algorithms decrease. This observation shows that our \paonline algorithm is more effective and the cost reduction is more significant for microgrids with high peak prices.

\subsection{The Performance of \paonline under Different Local Generation Capacities}\label{part:capa}
\textbf{Purpose.} At first glance, one may imagine that larger local generator leads to larger design space and thus larger cost reduction is expected.
However, as discussed in Sec.~\ref{part:criticalCapa}, this is not the case for online algorithms that do not have the complete future knowledge of price and demand. % If so, how will the local generator's capacity impact the performance? Answers to these question could be used as the guidelines for microgrid operators for deployment of as least as possible local generators.
We carry out an experiment to verify and elaborate the observation. For convenience, we define
%$\rho = \frac{C}{\max{e(t)}}$.
$\rho = C/\max{e(t)}$ as the ratio of local generation capacity over the peak net demand and change $\rho$ from $20\%$ to $100\%$. The result is shown in Fig.~\ref{fig:ReductionCapa}.

\textbf{Observations.} The results for \offline  and \poonline algorithms follow the intuition that more local capacity brings more cost reduction.
For \textsf{PA-Online}, however, we observe that the cost reduction increases when $\rho$ increases from $20\%$ to $60\%$, and degrades as $\rho$ continues to increase from $60\%$ to $100\%$. %This unexpected observation can be justified as follows.
%We remark that the experiment result is sensitive to multiple results, and generation capacity is only one of them.
%In fact, this counter intuitive phenomenon comes from two parts.
As we discussed in Sec.~\ref{part:criticalCapa}, there exists a critical local generation capacity $\tilde{C}$ beyond which the peak charge and the overall cost will not decrease further. In Fig.~\ref{fig:VPCapa}, we report the peak grid demand $\max{v(t)}$ versus $\rho$ just for \offline algorithm. Results show that the peak value of $v(t)$ does not decrease as $\rho$ increases from $60\%$ to $100\%$, evincing that $\tilde{C}$ is about $60\%$ of the maximum demand in this case. The online algorithm, however, is unaware of $\tilde{C}$. As discussed in Sec.~\ref{part:criticalCapa}, $\tilde{C}$ can be computed by solving problem \textbf{FS-PAED} in an offline manner.

The online algorithm, without knowing $\tilde{C}$ and with the tendency of reducing the peak charge by using more expensive local generation, will try to exploit the entire local generation capacity until the cost-benefit break even point is reached, which turns out to be less economic and deviate from the offline optimal. As a result, for the online algorithm, larger capacity may incur higher operating cost, as shown in Fig.~\ref{fig:ReductionCapa}.

This experiment, together with the discussions in Sec.~\ref{part:criticalCapa}, show that it is important for the microgrid operator to set the local generation capacity right at $\tilde{C}$ to cope with online algorithms to achieve maximum cost reduction. A possible way to set $\tilde{C}$ is to use the historical data as the input to the offline algorithm and obtain the critical capacity.
%Secondly, a limited capacity will force the microgrid to procure electricity from the grid earlier, especially for the lower layers, which can, to some extend, avoids this additional waste.

\subsection{The Performance of \textbf{RED}}

\textbf{Purpose.} In this part, we compare the empirical performance of the deterministic online algorithm \textbf{BED} and randomized online algorithm \textbf{RED} under different local capacities. The cost of \textbf{RED} is computed by running the algorithm 1000 times and taking the average. %We also compare \textbf{RED} with another randomized online algorithm, in which the probability distribution is chosen to be $\hat{f}(s) = \frac{e^s}{e-1},s\in [0,1]$, the one used in the classic ski rental problem. The results are shown in Fig~\ref{fig:RRand}.

\textbf{Observations.} Even though \textbf{RED} is better than \textbf{BED} in terms of competitive ratio, it is not always the case empirically because the competitive ratio only characterizes the performance in the worst case. As we can see, when $\rho$ is less than 80\%, \textbf{BED} outperforms \textbf{RED} while the other way around if $\rho$ is larger than 80\%. Furthermore, when $\rho$ increases from 80\% to 100\%, the performance of \textbf{BED} degrades drastically, while the cost reduction of \textbf{RED} almost remains the same. This observation indicates that, to ensure that \textbf{BED} has good performance, we need to carefully determine the local capacity but additional local capacity will not harm \textbf{RED} much, which can be viewed as another advantage of \textbf{RED}. %On the other hand, it is easy to see that \textbf{RED} outperforms the randomized online algorithm with $\hat{f}(s)$, which indicates the necessity of choosing the probability carefully.

\begin{figure}[t]
\centering
\begin{minipage}[!htbp]{0.48\columnwidth}
\includegraphics[scale=\myScl]{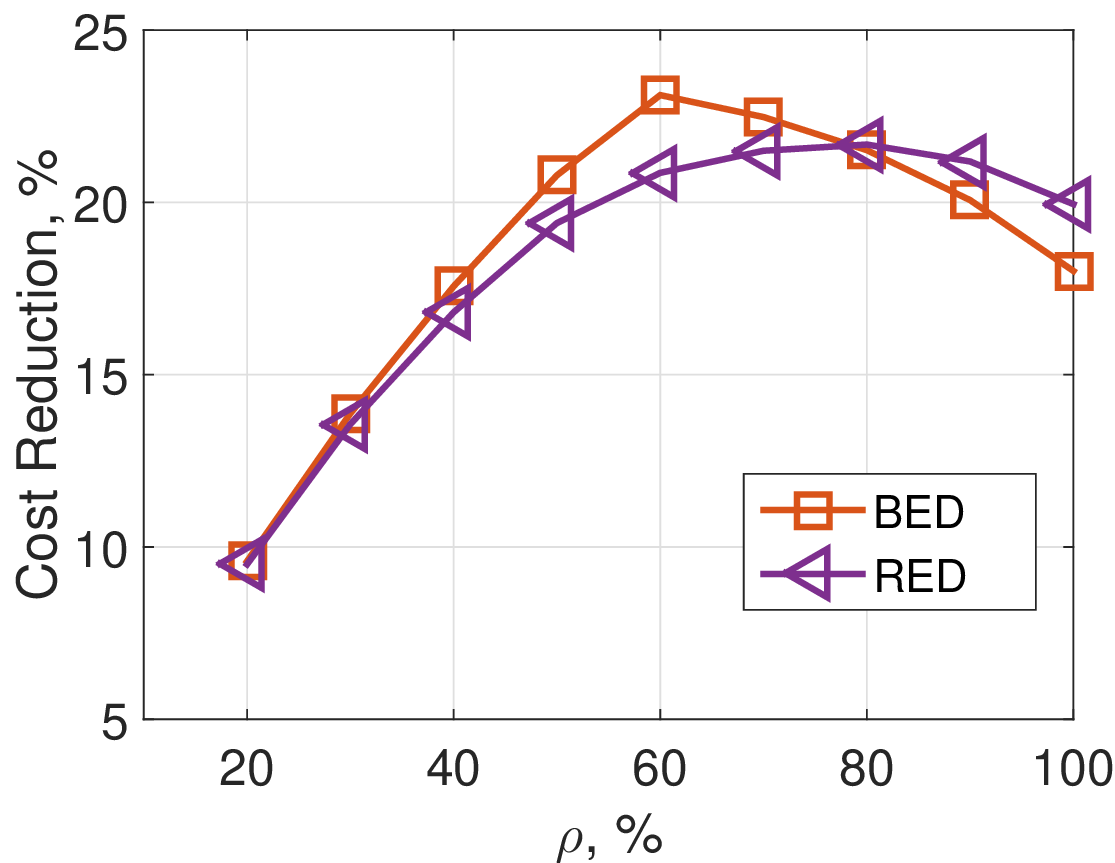}
\caption{Comparison of \textbf{RED} and \textbf{BED} under different local capacities}
\label{fig:RRand}
\end{minipage}
\begin{minipage}[!htbp]{0.50\columnwidth}
\includegraphics[scale=\myScl]{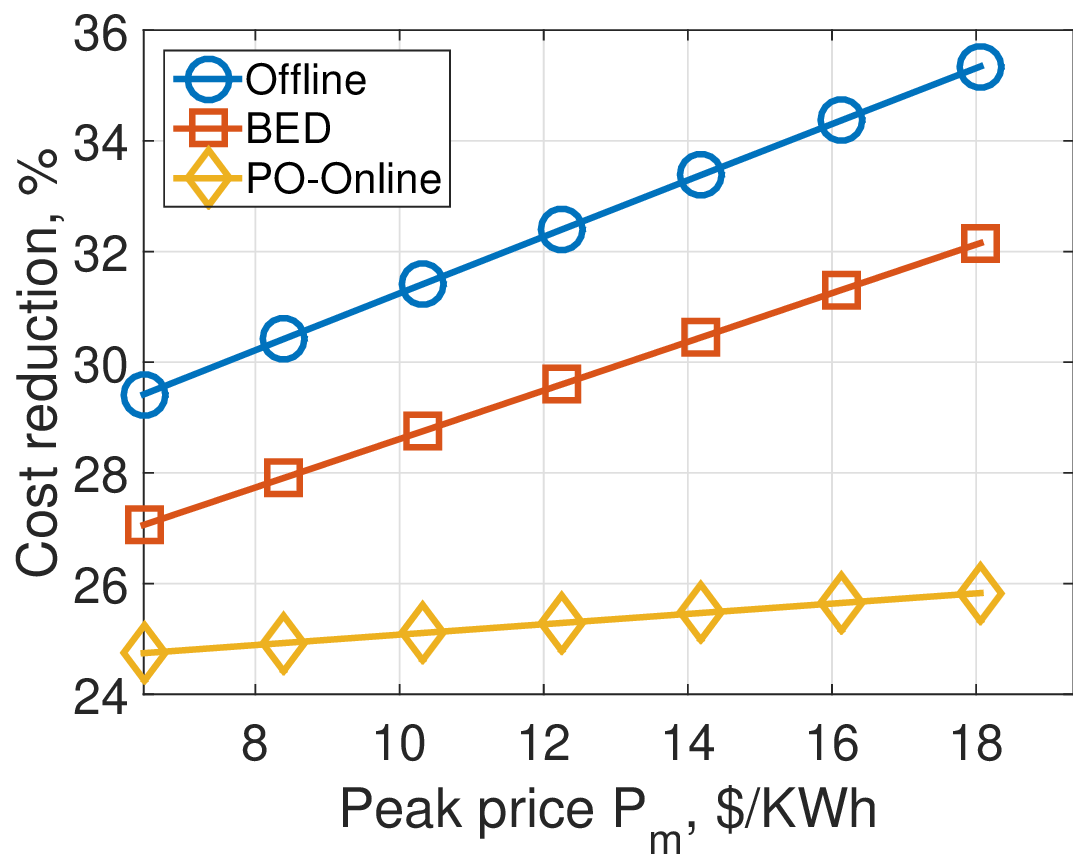}
\caption{Cost reduction with different peak price on a small microgrid testbed}
\label{fig:RMicrogrid}
\end{minipage}%
\end{figure}

\subsection{Empirical Evaluations Using Traces from a Real-world Small-scale Microgrid}
\textbf{Purpose.} In this simulation, we replace the previous trace with a new one, which is from a test-bed building at College of Engineering Center for Environmental Research and Technology of UC Riverside and spans three months from May to July. The building has 20 office rooms, 2 conference rooms, one large open area with cubicles, and 7 other miscellaneous rooms. The building HVAC system consists of 16 packaged rooftop units. In addition to its small scale, the building is connected to solar PV and several charging stations, both of which introduce additional demand uncertainties. As a result, the demand fluctuates more than the previous data set we use. %The purpose of this simulation is to show the necessity of peak-aware scheduling with more fluctuating demand and different peak prices.
The simulation result is shown in Fig~\ref{fig:RMicrogrid}.

\textbf{Observations.} On this new data set, the cost reduction is more significant (at least 40\% for the offline case) than the previous results and will increase with larger peak price $p_m$. This result indicates that peak-aware scheduling is more beneficial with more fluctuating demand and larger peak prices.

\subsection{The Impact of Ramping Constraints}\label{part:ramp}
\textbf{Purpose.} The experiment is devoted to explore the performance of our algorithm \algc for slow-responding generators. We firstly change the ramping constraint such that $\Gamma$ increases from $2$ to $5$ and evaluate the performance of the algorithms. We recall the meaning of $\Gamma$, as defined in~Sec.~\ref{sec:slow}, is that it takes $\Gamma$ slots for local generators to ramp up from zero to full capacity or down from full capacity to zero.
%In particular, we use OCO in \cite{narayanaswamy2012online} as the Peak-Oblivious online algorithm for additional comparison.
%We compare the performances of the peak-aware offline algorithm, our proposed peak-aware online algorithm (\textsf{NRBF}), peak-oblivious offline algorithm, and peak-oblivious online algorithm  \cite{narayanaswamy2012online} (based on OCO approach).
%\footnote{This comparison may be a little unfair for OCO \cite{narayanaswamy2012online} since it is designed with different purposes. }
The result is demonstrated in Fig.~\ref{fig:RRamp}.

Secondly, we relax the assumption that we can perfectly predict near future information. We add a zero mean gaussian noise to the net demand as the predicted input for our algorithm. Note that we can always satisfy unexpected demand by purchasing electricity from the external grid. We evaluate the performance of \algc with standard deviation of the gaussian noise increasing from $0\%$ to $10\%$ of the actual demand. We use $\Gamma = 4$ for the experiment and \algc utilizes 3-slot looking ahead demand and price information. The simulation results are shown in Fig.~\ref{fig:CRRamp}.

\textbf{Observations.} From Fig.~\ref{fig:RRamp}, we observe that all cost reductions decrease as ramping constraints are more strict, while the performance of \algc is always close to the offline optimal. This shows the effectiveness of using limited prediction in combating the difficulty in online scheduling caused by ramping constraints. Moreover, Fig.~\ref{fig:CRRamp} shows that the prediction error will degrade the performance of \algc. However, the performance is still significantly better than the peak-oblivious scheduling.

%All the above verify that \algc is a very effective online algorithm for $\textbf{PAED}$ and is applicable in practice.
%Another interesting observation is, the performance of \algc  is closer to the offline algorithm if $\omega$ is larger, which is also indicated by the competitive ratio change in Fig.~\ref{fig:CRRamp}. This is because \algc requires more future information as $\omega$ is larger. 

%\vspace{-0.5\baselineskip}
\section{Related Work}\label{sec:relatedwork}
%Economic Dispatching, as an important part of energy generation scheduling problem, determines the output level of generators to meet the real-time demand without alternating the startup and shutdown schedule of generators, which is determined by the Unit Commitment(UC) problem. Comprehensive survey on UC and ED are available in \cite{UC} and \cite{gaing2003particle}.

Microgrid is attracting substantial attention from both academic and industrial communities due to its economic and environmental benefits, evidenced by a number of real-world pilot microgrid projects~\cite{barnes2007real}.

With the penetration of renewable energy in microgrids, conventional economic dispatching approaches based on accurate demand prediction for power grid \cite{gaing2003particle} are not applicable as the net demand inherits substantial uncertainty from the renewable generation and is hard to predict accurately. Online algorithm design is advocated by researchers to offer a paradigm-shrift alternative. Online convex optimization \cite{narayanaswamy2012online}, Lyapunov optimization \cite{huang2013adaptive}, and competitive analysis \cite{minghua_sigmetrics} are the main approaches adopted for online energy generation scheduling in microgrid. The authors in \cite{minghua_sigmetrics} study the unit commitment and economic dispatching problems of microgrid under the volume charging model. Our work considers economic dispatching under both the peak charging and volume charging model.
%Thus, this work not only addresses the physical constraints of generators, but also tries to utilize the local generators to reduce the peak charge of grid by proposing a peak-aware scheduling.
%, but adopt a peak-oblivious cost model, thus the results do not apply to the problem we study in this paper.

%The authors in \cite{narayanaswamy2012online} designed ED algorithms for microgird by online convex optimization framework. The authors in \cite{huang2013adaptive} applied Lyapunov optimization framework to design algorithm for microgrid, with consideration of energy storage. In \cite{minghua_sigmetrics}, the authors investigated the microgirds by considering electricity and heat co-generation and designed an online algorithm CHASE for energy generation scheduling with provable performance guarantee.

The cost minimization problem based on real-world peak charging scheme has been considered for microgird scenario in \cite{mishra2013scaling}, by utilizing Energy Storage Systems to cut off the peak.
In contrast, our work tackles the problem using local generators to shave the peak.  The cost minimization with the same pricing mechanism taken into account is also studied for data centers in \cite{xu2013reducing,wang2014hierarchical}, for EV charging in \cite{Minghua_allerton}, and for content delivery in~\cite{adler2011algorithms}. For fast-responding generator scenario, the economic dispatching problem we study in this paper can be considered as a generalization of the classic  Bahncard problem \cite{bahncard}. The Bahncard problem and its solutions have also found application in the instance acquisition problem of cloud computing~\cite{reserve}.

%, but the solutions $A_1, A_{f^*}$ here are independently developed. The randomized input Eq~\eqref{equ:random_z} to prove the optimal randomized online algorithm in this paper is new.

%Ski rental problem is well studied in \cite{karlin1988competitive,karlin1994competitive} and extended in \cite{lotker2008ski,ai2014multi}, etc. Recently it finds new applications in \cite{khanafer2013constrained,dong2014cost}, etc. The mathematic model of $\textbf{FS-PAED}^k$ in this paper is similar to that proposed in \cite{lotker2008ski} in the sense that the customer needs to rent at another lower rate after buying, but in our model, this renting rate is dynamic and unknown to the customer.

%\vspace{-0.5\baselineskip}
\section{Conclusion and Future Work}\label{sec:conclusion}

In this paper, we devised peak-aware online economic dispatching algorithms for microgrids, with peak charging model taken into account. In the fast-responding generator scenario, we developed both deterministic and randomized online algorithms with best possible competitive ratios following a divide-and-conquer approach. Our results not only characterized the fundamental price of uncertainty for the problem, but also served as a building block for designing online algorithms for the slow-responding generator scenario, where we proposed to tackle the ramping constraints using a limited look-ahead window. In addition to sound theoretical performance guarantees, the empirical evaluations based on real-world traces also corroborated our claim on the importance of peak-awareness in scheduling.
%Overall, we believe that the results in this paper are important for the energy scheduling problems in microgrids.

An interesting future direction is to study the microgrid economic dispatching problem under accurate or noisy prediction of future demand and renewable generation within a limited look-ahead window.
%
%further study how to leverage the future knowledge to handle the ramping constraints under more realistic conditions, for example, when the looking ahead window is not sufficiently large or the prediction is contaminated with noise.

%A future direction is to combine the local generating units and energy storage devices efficiently to do the economic dispatching, with respect to the peak-based charging scheme. And also, it is interesting to further study how to leverage the future knowledge to handle the ramping constraints under more realistic conditions, for example, when the looking ahead window is not sufficiently large or the prediction is contaminated with noise.

\section*{Acknowledgement}
The authors would like to thank Qi Zhu for the discussions on peak charging in the initial stage of the study.
The first author wants to thank Shaoquan Zhang for proofreading the paper.
The work described in this paper was supported by National Basic Research Program of China (Project No. 2013CB336700) and the University Grants Committee of the Hong Kong Special Administrative Region, China (General Research Fund Project No. 14201014 and Theme-based Research Scheme Project No. T23-407/13-N).

\scriptsize{
\bibliographystyle{abbrv}
\bibliography{ref}
}
\end{document}